\documentclass[10 pt, letterpaper]{article}

\usepackage{amssymb}
\usepackage{amsmath}
\usepackage{graphicx}
\usepackage{times}
\usepackage[T1]{fontenc}

\newtheorem{assumption}{Assumption}
\newtheorem{proposition}{Proposition}
\newtheorem{proof}{Proof}
\newcommand\highertop{\rule{0pt}{3.1ex}}

\begin{document}
\date{}

\title{Automatic crosswind flight of tethered wings\\ for airborne wind energy:\\ modeling, control design and experimental results
\thanks{This manuscript is a preprint of a paper submitted for possible publication on the IEEE Transactions on Control Systems Technology and is subject to IEEE Copyright. If accepted, the copy of
record will be available at \textbf{IEEEXplore} library: http://ieeexplore.ieee.org/.}
\thanks{This research has received funding from the California Energy Commission under the EISG grant n. 56983A/10-15 ``Autonomous flexible wings for high-altitude wind energy generation'', and from the European Union Seventh Framework Programme (FP7/2007-2013) under grant agreement n. PIOF-GA-2009-252284 - Marie Curie project ``Innovative Control, Identification and Estimation Methodologies for Sustainable Energy Technologies''. The authors acknowledge SpeedGoat$^\circledR$'s Greengoat program.}}
\author{L. Fagiano\thanks{Corresponding author: fagiano@control.ee.ethz.ch.},
 A. U. Zgraggen,
 M. Morari and M. Khammash
\thanks{L. Fagiano, A. Zgraggen and M. Morari are with the Automatic Control Laboratory, Swiss Federal Institute of Technology, Zurich, Switzerland. M. Khammash is with the Department of Biosystems Science and Engineering, ETH Zurich, Switzerland. L. Fagiano and M. Khammash are also with the Dept. of Mechanical Engineering, University of California at Santa Barbara, CA, USA.}}
\maketitle

\begin{abstract}
An approach to control tethered wings for airborne wind energy is proposed. A fixed length of the lines is considered, and the aim of the control system is to obtain figure-eight crosswind trajectories. The proposed technique is based on the notion of the  wing's ``velocity angle'' and, in contrast with most existing approaches, it does not require a measurement of the wind speed or of the effective wind at the wing's location. Moreover, the proposed approach features  few  parameters, whose effects on the system's behavior are very intuitive, hence simplifying tuning procedures. A simplified model of the steering dynamics of the wing is derived from first-principle laws, compared with experimental data and used for the control design. The control algorithm is divided into a low-level loop for the velocity angle and a high-level guidance strategy to achieve the desired flight patterns. The robustness of the inner loop is verified analytically,  and the overall control system is tested experimentally on a small-scale prototype, with varying wind conditions and using different wings.

\end{abstract}

\section{Introduction}\label{S:intro}

Airborne wind energy systems aim at harnessing the wind blowing up to 1000 m above the ground, using tethered wings or aircrafts. In recent years, an increasing number of researchers in academia and industry started to investigate this idea and to develop concepts of airborne wind energy generators, see e.g. \cite{Makani,skysails,ampyx,windlift,kitenergy,enerkite,CaFM07,IlHD07,CaFM09c,TBSO11,BaOc12,VeGK12} as well as \cite{FaMi12} for an overview.\\ 
In several concepts of airborne wind energy generators that are currently being developed, a tethered flexible wing is controlled to fly fast in crosswind conditions, i.e. roughly perpendicular to the wind flow \cite{Loyd80}, and the traction forces acting on the lines are converted into electricity using mechanical and electrical equipments installed on the ground \cite{skysails,kitenergy,enerkite,BaOc12}. In particular, on the ground the wing's lines are wound around one or more winches, linked to electric generators. Energy is obtained by continuously performing a two-phase cycle, composed by a \emph{traction phase}, during which the lines are unrolled under high traction forces and the generators, driven by the rotation of the winches, produce electricity, and by a subsequent \emph{passive phase}, when the electric generators act as motors, spending a fraction of the previously generated energy to recoil the lines.\\
The automatic control of the wing is a key aspect of airborne wind energy. In traction phases, the aim is to make the wing fly along figure-eight paths, which yield the highest traction forces while preventing line twisting. This control problem involves fast, nonlinear, unstable time-varying dynamics subject to hard operational constraints and external disturbances. Several contributions by various research groups and companies worldwide have recently appeared in the literature, see e.g. \cite{IlHD07,WiLO08,CaFM09c,FaMP11,ErSt12,BaOc12}. Most of the presented approaches \cite{CaFM07,IlHD07,WiLO08,CaFM09c,phd_thesis_fagiano} are based on a nonlinear point-mass model of the system, derived on the basis of first-principle laws of mechanics and aerodynamics, and they rely on the use of advanced nonlinear control design techniques. Several proposed techniques employ a reference trajectory derived off-line on the basis of the considered model, usually computed in order to  be optimal in terms of generated power \cite{IlHD07,WiLO08,BaOc12}. Then, approaches like tracking Model Predictive Control (MPC) \cite{IlHD07} or adaptive control-Lyapunov techniques \cite{BaOc12} are used to track the reference path. The use of MPC with an economic cost function, i.e. without the need to pre-compute a reference path, has been also proposed, see e.g. \cite{CaFM09c},
again exploiting a point mass model. While the mentioned contributions represent fascinating applications of constrained and optimal nonlinear control methods, their use in a real system appears to be not trivial, due to the discrepancies between the employed simplified model and the real dynamics of flexible wings, the need to solve complex nonlinear optimization problems in real-time, finally the need to measure the wind speed and direction at the wing's flying altitude.\\ In contrast with the mentioned works, in a recent contribution, concerned with the control of large kites for seagoing vessels \cite{ErSt12}, a simpler dynamical model has been proposed and used for control design. Such a model has been justified by means of measured data, and the designed control system has been tested experimentally, thus showing the practical applicability of the approach. It has to be noted that a similar model has been considered also in \cite{BaOc12}, where it has been justified by a priori assumptions.
Hence, neither \cite{BaOc12} nor \cite{ErSt12} provided an explicit link between the model considered in the control design and the wing's characteristics, like area, efficiency and mass.  The control approach proposed in \cite{ErSt12} is composed by an inner control loop that computes the  input needed to obtain a desired reference heading of the wing, and an outer control loop that computes the reference heading according to a bang-bang-like strategy. The inner control loop is a quite sophisticated model following approach and it needs the measure of the effective wind speed at the wing's altitude, obtained from an onboard anemometer.\\
In the described context, we present here new contributions in the field of control of flexible tethered wings for airborne wind energy systems with ground-level generators. We focus on the problem of controlling the wing in order to fly along figure-eight paths in crosswind conditions. First, we consider a simplified model based on the notion of ``velocity angle'' of the wing, similar in form to the one proposed by \cite{BaOc12} and \cite{ErSt12}, and we derive an explicit link between the model's parameters and the system's characteristics. We show the validity of such a model as compared to experimental data collected with a small-scale prototype.
This result provides a definitive assessment of the considered control-oriented model for tethered wings, thus bridging the gap between theoretical equations and experimental evidence.
As a second contribution,  we present a new control algorithm for tethered wings, based on the derived simplified model. Differently from \cite{IlHD07,WiLO08,BaOc12}, the approach does not employ pre-computed paths based on a mathematical model of the system, thus avoiding issues related to model mismatch and to the actual feasibility of the employed reference path for the (uncertain and time-varying) dynamics of a real system. Moreover, an estimate or measure of the wind speed at the wing's location, as considered e.g. in \cite{IlHD07,WiLO08,CaFM09c}, is not needed in our approach, nor is a measure of the effective wind speed aligned with the wing's longitudinal body axis, used in \cite{ErSt12}, but only a rough estimate of the wind direction with respect to the ground. The structure of the controller is similar to the one proposed in \cite{ErSt12}, but the controllers employed for the inner and outer control loops are different. In particular, the inner controller is a simple static gain, while the
outer one is  given by a switching strategy based on the wing's position. The proposed control approach involves few parameters, that can be tuned in an intuitive way. By exploiting the above-mentioned results pertaining to the derivation of the control-oriented model, we assess the robustness of the inner control loop analytically, against a wide range of operating conditions in terms of wind speed and wing characteristics. We then present the experimental results obtained by  testing the approach on a small-scale prototype realized at the University of California, Santa Barbara, in different wind conditions and using different wings. The paper is organized as follows. section \ref{S:System_description} describes the considered layout and the derivation of the simplified model for the velocity angle dynamics. The control design is presented in section \ref{S:control_design} and experimental results are given in section \ref{S:results}. Finally, conclusions are drawn in section \ref{S:conclusion}.

\section{System description and model equations}\label{S:System_description}
\subsection{System layout}\label{SS:yo_yo_descr}
We consider a flexible wing, or power kite, connected by three lines to a ground unit (GU).
This setup corresponds to a prototype built at the University of California, Santa Barbara, shown in Fig. \ref{F:proto}.
 \begin{figure}[hbt]
 \centerline{
 \includegraphics[bbllx=13mm,bblly=59mm,bburx=196mm,bbury=284mm,width=7cm,clip]{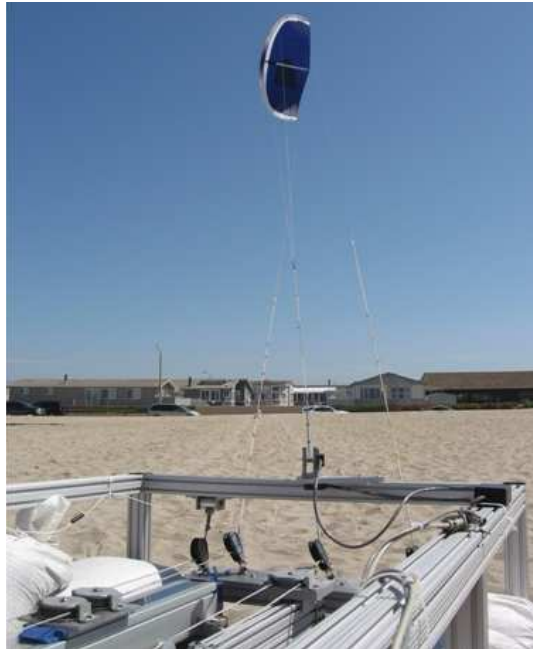}}
\caption{Small-scale prototype built at the University of California, Santa Barbara, to study the control of tethered wings for airborne wind energy.}\label{F:proto}
 \end{figure}
In normal flight conditions, the wing's trajectory evolves downwind with respect to the GU. For simplicity, we assume that the nominal wind direction (i.e. neglecting turbulence and small, zero-mean deviations) is aligned with the longitudinal symmetry axis of the GU, denoted by $X$. This condition can be achieved by properly orienting the GU, exploiting a measure or estimate of the nominal wind direction. With this assumption in mind, our control approach employs the feedback of the wing's position relative to the GU to obtain crosswind trajectories, i.e. flying paths that are downwind and symmetric with respect to $X$ axis (i.e. the wind direction). From our experimental results, it turns out that misalignments of about $\pm30^\circ$ between the wind and the GU do not change significantly the obtained paths relative to the GU, but indeed give place to lower forces on the lines, as expected from the theory (see e.g. \cite{FaMP11}). The $X$ axis, together with the $Z$ axis being perpendicular to
the ground and pointing upwards and with the $Y$ axis to complete a right-handed system, forms the inertial frame $G\doteq(X,Y,Z)$, centered at the GU (see Fig. \ref{F:system_1}). By considering a fixed length of the lines, denoted by $r$, the wing's trajectory is thus confined on a quarter sphere, given by the intersection of a sphere of radius $r$ centered at the GU's location and the planes ${(x,y,z)\in\mathbb{R}^3:x\geq0}$ and $z\geq0$.
\begin{figure}[hbt]
\centerline{
\includegraphics[bbllx=79mm,bblly=125mm,bburx=201mm,bbury=218mm,width=8cm,clip]{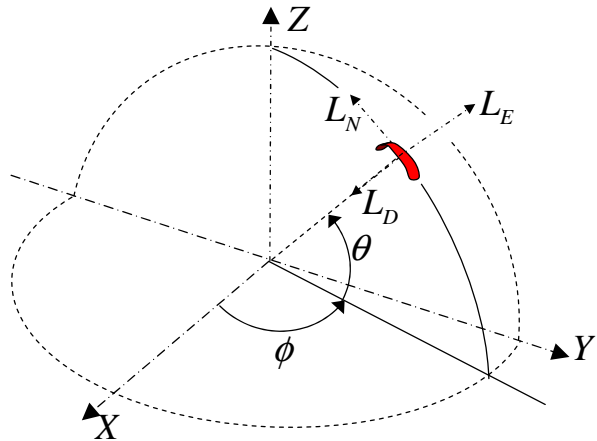}
} \caption{Reference system $G=(X,Y,Z)$, wind window (dashed lines), variables $\theta,\,\phi$, and local north, east and down $(L_N,L_E,L_D)$ axes.}\label{F:system_1}
\end{figure}
Such a quarter sphere is commonly named ``wind window'', see Fig. \ref{F:system_1} (dashed lines).\\
The two lateral lines, named steering lines, are linked to the back tips of the wing (see Fig. \ref{F:proto}) and they are used to influence its trajectory: a shorter left steering line with respect to the right one impresses a left turn to the wing (i.e. a counter-clockwise turn as seen from the GU), and vice-versa.  The center line, named power line, splits into two lines connected to the front of the wing (or leading edge) and sustains about 70\% of the generated load. The GU is designed to support the forces
acting on the lines and it is equipped with actuators, able to achieve a desired difference of steering lines' length. In the considered prototype, a single motor, together with a linear motion system (visible in the lower-left corner of Fig. \ref{F:proto}), is able to change the difference of length of the steering lines. In this work, we consider a fixed length of the lines of $r=30\,$m and we focus on the problem of designing a controller able to make the wing fly along ``figure-eight'' crosswind paths, which maximize the generated forces. The use of a fixed lines' length does not limit the significance of our results, since the problem of crosswind path control can be decoupled from the problem of controlling the lines' reeling. Indeed it has been shown that the optimal operation of ground-based airborne wind energy generators is achieved with a constant line speed (equal to approximately one third of the wind speed to maximize power production, see e.g. \cite{FaMP11}). The settings considered here can be seen as a particular case of constant line speed, equal to zero. Other approaches in the literature follow a similar strategy of using two separate control systems, one for the wing guidance and one for the line unrolling (see e.g. \cite{CaFM09c,BaOc12,ErSt12}). Rather, the use of such short lines makes the control problem more challenging, since the flight paths must be contained in a small area, in order to avoid contact with the ground or aerodynamic stall occurring at the border of the wind window.\\
The wing is equipped with onboard sensors and a radio transmitter; the receiver and other sensors are installed on the GU. 
The available sensors, together with suitable filtering algorithms, provide accurate real-time estimates of the wing's position and velocity vector, to be used for feedback control. For the sake of simplicity, in this paper the feedback variables are considered to be measured with no measurement noise; the interested reader is referred to \cite{FHBK12_arxiv} for details on the design of suitable filtering algorithms for the described setup. In the next section, we briefly recall a point-mass dynamical model of the system and introduce the notion of velocity angle of the wing. The latter represents one of the  feedback variables in our control approach.

\subsection{Model equations}\label{SS:model}

We consider the point-mass model already employed in previous works (see e.g. \cite{IlHD07,CaFM09c,BaOc12,FaMP09} and references therein) and we adapt it to a slightly different reference system, in order to introduce the wing's velocity angle. The latter is one of the feedback variables we use and it represents a novelty with respect to the mentioned previous approaches. For the sake of completeness, we recall here the main equations of the point-mass model, since they are instrumental to prove a theoretical result concerned with the steering dynamics of tethered wings.\\
By considering the fixed line length $r$, the wing's position can be expressed in the inertial frame $G$ by using the spherical coordinates $\theta(t),\,\phi(t)$ as (see Fig. \ref{F:system_1}):
\begin{equation}
\label{eqn:rKite}
_{G}\vec{p}(t) = \left[
\begin{array}{c}
r\cos{(\phi(t))}\cos{(\theta(t))}\\
r\sin{(\phi(t))}\cos{(\theta(t))}\\
r\sin{(\theta(t))}
\end{array}\right],
\end{equation}
where $t$ is the continuous time variable. 
In (\ref{eqn:rKite}) and throughout the paper, the subscript letter in front of vectors (e.g. $_{G}\vec{p}(t)$) denotes the reference system considered to express the vector components.\\
We define also a non-inertial coordinate system $L\doteq (L_N,L_E,L_D)$, centered at the wing's position (also depicted in Fig. \ref{F:system_1}). The $L_N$ axis, or local north, is tangent to the wind window and points towards its zenith.
The $L_D$ axis, called local down, points to the center of the sphere (i.e. the GU), hence it is perpendicular to the tangent plane to the wind window at the wing's location. The $L_E$ axis, named local east, forms a right hand system and spans the tangent plane together with $L_N$. We note that the system $L$ is a function of the wing's position only, and it is different from the local systems used in previous works (see e.g. \cite{CaFM09c} and the references therein), due to the different definition of angle $\theta$. A vector in the $L$ system can be expressed in the $G$ system by means of the following rotation matrix:
 \begin{equation}\small
 \begin{array}{l}
  R =\\\left(\begin{array}{ccc}
         -\cos{(\phi(t))}\sin{(\theta(t))} & -\sin{(\phi(t))} & -\cos{(\phi(t))}\cos{(\theta(t))}\\
         -\sin{(\phi(t))}\sin{(\theta(t))} &  \cos{(\phi(t))} & -\sin{(\phi(t))}\cos{(\theta(t))}\\
          \cos{(\theta(t))}           &  0          & -\sin{(\theta(t))}
        \end{array}\right)\nonumber
        \end{array}\normalsize
 \end{equation}
By applying Newton's law of motion to the wing in the reference $L$ we obtain:
\begin{subequations}
\begin{align}
   \ddot{\theta}(t)=\dfrac{_L\vec{F}(t)\cdot \vec{e}_{L_N}(t)}{rm} - \sin{(\theta(t))}\cos{(\theta(t))}\dot{\phi}^2(t)\label{eqn:EqnMot1}\\
   \ddot{\phi}(t)=\dfrac{_L\vec{F}(t)\cdot \vec{e}_{L_E}(t)}{rm\cos{(\theta(t))}} + 2\tan{(\theta(t))}\dot{\theta}(t)\dot{\phi}(t),
\label{eqn:EqnMot2}
\end{align}
\end{subequations}
where $m$ is the mass of the wing. In \eqref{eqn:EqnMot1}-\eqref{eqn:EqnMot2} and throughout the paper we denote unit vectors by ``$\vec{e}\,$'' followed by a subscript indicating the related axis, e.g. $\vec{e}_{L_N}(t)$ denotes the unit vector of the $L_N$ axis. The force $_L\vec{F}(t)$ consists of
contributions from the gravity force $_L\vec{F}_\text{g}(t)$ and aerodynamic force $_L\vec{F}_\text{a}(t)$.
Vector $_L\vec{F}_\text{g}(t)$ can be computed as:
\begin{equation}\label{eqn:Fgrav}
{}_L\vec{F}_\text{g}(t) =\left[
\begin{array}{c}
-mg\cos{(\theta(t))}\\
0\\
mg\sin{(\theta(t))}\end{array}
\right],
\end{equation}
where $g$ is the gravity acceleration.
The aerodynamic force is given by the contributions of the lift and drag generated by the wing and of the drag induced by the cable.
These forces depend on the the effective wind, $\vec{W}_e(t)$, computed as:
\begin{equation}\label{E:eff_wind}
\vec{W}_e(t)=\vec{W}(t)-\vec{v}(t),
\end{equation}
where $\vec{W}(t)$ is the wind relative to the ground and $\vec{v}(t) \doteq \frac{d}{dt}\vec{p}(t)$ is the wing velocity vector, which can be expressed in the $L$ frame as
\begin{equation}\label{eqn:vP}
_{L}\vec{v}(t)= \left[\begin{array}{c}
r\dot{\theta}(t)\\
r\cos{\left(\theta(t)\right)}\dot{\phi}(t)\\
0
\end{array}\right].
\end{equation}
Then, the  aerodynamic force $_L\vec{F}_\text{a}(t)$ can be computed as (see e.g. \cite{CaFM09c}):
\begin{subequations}\label{E:aero_force}
\begin{align}
 \vec{F}_\text{a}(t)=&\frac{1}{2}\rho C_L(t)A|\vec{W}_e(t)|^2\vec{z}_w(t) + \label{E:aero_force1}\\
                 &\frac{1}{2}\rho C_D(t)A|\vec{W}_e(t)|^2\vec{x}_w(t) + \label{E:aero_force2}\\
                 & \frac{1}{8}\rho C_{D,l}A_l\cos{\left(\Delta\alpha(t)\right)}|\vec{W}_e(t)|^2\vec{x}_w(t)\label{E:aero_force3}\\
                =& \frac{1}{2}\rho C_L(t)A|\vec{W}_e(t)|^2\vec{z}_w(t) + \nonumber\\
                 &\frac{1}{2}\rho\underbrace{\left(C_D(t)+\frac{C_{D,l}A_l\cos{\left(\Delta\alpha(t)\right)}}{4A}\right)}_{C_{D,eq}(t)}
                               A|\vec{W}_e(t)|^2\vec{x}_w(t)\label{E:aero_force4}.
\end{align}
\end{subequations}
In \eqref{E:aero_force}, the contributions \eqref{E:aero_force1}-\eqref{E:aero_force2} are, respectively, the lift and drag forces generated by the wing, while \eqref{E:aero_force3} is the drag induced by the lines. $C_L(t)$ and $C_D(t)$ are the aerodynamic lift and drag coefficients of the wing, $C_{D,l}$ is the drag coefficient of the lines,
$A$ is the reference area of the wing, $A_l$ is the reference area of the lines, $\rho$ is the air density, and $\vec{x}_w(t)$
and $\vec{z}_w(t)$ are the directions of the drag and lift forces, respectively. The parameter $C_{D,eq}(t)$ is called the equivalent aerodynamic drag coefficient, since it accounts for the drag of both the wing and the lines. We note that the aerodynamic coefficients are considered as time-varying parameters here, since they depend on the wing's angle of attack, which in turns changes in time as a function of the flight conditions. The variable
$\Delta\alpha(t)$ is the angle between the effective wind vector $\vec{W}_e(t)$ and the tangent plane to the wind window at the wing's location. The vectors $\vec{x}_w(t)$ and $\vec{z}_w(t)$, defining the directions of the lift and drag forces, depend on the direction of the effective wind and on the roll angle
$\psi(t)$ of the wing. In particular, $\vec{x}_w(t)$ points in the direction of the effective wind $\vec{W}_e(t)$, while $\vec{z}_w(t)$ is perpendicular to $\vec{x}_w(t)$ and
to a further vector, denoted by $\vec{e}_t(t)$, which points from the right tip of the wing to the left one, as seen from the GU (see e.g. \cite{CaFM09c} for a
formal definition).\\

Vectors $\vec{x}_w(t)$ and $\vec{z}_w(t)$ can be expressed in the $L$ frame as:
\begin{subequations}\label{E:aero_vectors}\small
\begin{align}
 {}_L\vec{x}_w(t)=\begin{pmatrix}
		           -\cos{(\xi(t))}&-\sin{(\xi(t))}& 0\\
		           -\sin{(\xi(t))}& \cos{(\xi(t))}& 0\\
		            0         & 0          &-1
		          \end{pmatrix}\cdot		          \begin{pmatrix}\cos{(\Delta\alpha(t))}\\0\\\sin{(\Delta\alpha(t))}\end{pmatrix}\label{eqn:Lxw}\\
 {}_L\vec{z}_w(t)=\begin{pmatrix}
		           -\cos{(\xi(t))}&-\sin{(\xi(t))}& 0\\
		           -\sin{(\xi(t))}& \cos{(\xi(t))}& 0\\
		            0          & 0         &-1
		          \end{pmatrix}\label{eqn:Lzw}\cdot\\
	            \begin{pmatrix}
                           -\cos{(\psi(t))}\cos{\left(\eta(t)\right)}\sin{(\Delta\alpha(t))}\\
                           \cos{(\psi(t))}\sin{\left(\eta(t)\right)}\sin{(\Delta\alpha(t))}+\sin{(\psi(t))}\cos{(\Delta\alpha(t))}\\
                           \cos{(\psi(t))}\cos{\left(\eta(t)\right)}\cos{(\Delta\alpha(t))}
                          \end{pmatrix}\nonumber.
\end{align}\normalsize
\end{subequations}
In \eqref{E:aero_vectors},  $\eta(t)$ is given by (see e.g. \cite{ArRS09}):
\begin{equation}\label{eqn:eta}
 \eta(t) = \arcsin{\left(\tan{\left(\Delta\alpha(t)\right)}\tan{(\psi(t))}\right)},
\end{equation}
and $\psi(t)$ is a function of the steering input, $\delta(t)$:
\begin{equation}
 \psi(t) =\arcsin{\left(\frac{\delta(t)}{d_s}\right)}\label{eqn:psi}
\end{equation}
where $d_s$ is the wing span. Finally, $\xi(t)$ is the heading angle of the wing, and it is computed as the angle between the local north $L_N$ and the effective wind  $\vec{W}_e(t)$ projected on the $(L_N,L_E)$ plane:
\begin{equation}
 \xi(t)= \arctan{\left(\frac{\vec{W}_e(t)\cdot\vec{e}_{L_E}(t)}{\vec{W}_e(t)\cdot\vec{e}_{L_N}(t)}\right)}\label{eqn:Beta}.
\end{equation}
In \eqref{eqn:Beta}, the four-quadrant version of the arc tangent function shall be used, such that $\xi(t)\in[-\pi,\,\pi]$. The assumption underlying equation \eqref{eqn:Beta} for the computation of the wing's heading is that the wing's longitudinal body axis is always contained in the plane spanned by vectors  $\vec{W}_e(t)$ and $\vec{p}(t)$.\\
Equations \eqref{eqn:EqnMot1}-\eqref{eqn:Beta} give an analytic
expression for the point-mass model of the wing, with four states ($\theta(t),\,\phi(t),\,\dot{\theta}(t),\dot{\phi}(t)$), one manipulated input ($\delta(t)$) and three exogenous inputs (the components of vector $\vec{W}(t)$). Such a model has been widely used in the literature on control design for airborne wind energy applications, see e.g. \cite{CaFM07,IlHD07,CaFM09c,BaOc12}. Several existing approaches are based on nonlinear model predictive control techniques \cite{IlHD07,CaFM09c}, leading to quite complex multivariable controllers that rely on the feedback of the four states. In the section \ref{S:control_model}, we show how a simpler model, which we will use for our control design, can be derived from the dynamics \eqref{eqn:EqnMot1}-\eqref{eqn:Beta}. The variable involved in such control-oriented model is the velocity angle $\gamma(t)$ of the wing, defined as:
\begin{equation}
 \gamma(t) \doteq \arctan{\left(\frac{\vec{v}(t)\cdot\vec{e}_{L_E}(t)}{\vec{v}(t)\cdot\vec{e}_{L_N}(t)}\right)}
         =  \arctan{\left(\dfrac{\cos{(\theta(t))}\dot{\phi}(t)}{\dot{\theta}(t)}\right)}.\label{eqn:Gamma}
\end{equation}
The angle $\gamma(t)$ is thus the angle between the local north $\vec{e}_{L_N}(t)$ and the wing's velocity vector $\vec{v}(t)$. This variable is particularly suited for feedback control,
since it describes the flight conditions of the wing with just one scalar: as an example,  if $\gamma =0$ the wing is moving upwards towards the
zenith of sphere, if $\gamma = \pi/2$ the wing is moving parallel to the ground towards the local east, finally if $\gamma = \pi$ the wing
is flying towards the ground. Hence, the time derivative $\dot{\gamma}$ defines how fast the wing is being steered while flying in the wind window. 
Similarly to \eqref{eqn:Beta}, also in \eqref{eqn:Gamma} the four-quadrant version of the arc tangent function shall be used, such that $\gamma(t)\in[-\pi,\,\pi]$. 

\subsection{Input model}\label{SS:input_model}

The prototype used for our test flights features two attachment points for the steering lines on the GU, left and right. These attachment points
are separated by a distance $d$ (see Fig. \ref{F:proto}, in the lower part, where the attachment points with swaying pulleys are visible). When the wing's lines are not aligned with the $X$ axis, this distance induces an equivalent steering deviation. We call such deviation ``geometric input'', $\delta_\text{g}$, since its value depends on the geometry of the attachment points and on the ($\theta,\phi$) position of the wing in the wind window. Hence, the overall steering input acting on the wing is:
\begin{equation}\label{eqn:input_total}
\delta(t)=\delta_\text{u}(t)+\delta_\text{g}(t),
\end{equation}
where $\delta_\text{u}(t)$ is the input issued by the control system, i.e. the difference of length of the steering lines (right minus left) obtained by changing the position of the linear motion system on the GU. In this section, we derive an expression to compute the geometric input as a function of $\theta(t)$, $\phi(t)$ and $d$. 
Consider the GU and the wing as seen from above, such that the trace of the wind window at the wing's height is a semicircle of radius  $r\cos{(\theta)}$ and, for fixed $\theta$, the wing position is univocally identified by $\phi$. Let us consider first a situation in which both wing tips lie on the tangent plane to the wind window at the current wing position $(\theta,\phi)$. Since the steering of the wing is essentially induced by a roll motion, by which the wing tips are moved away from the tangent plane, we call this orientation ``neutral configuration'' (see Fig. \ref{fig:PhiThetaInputDist}, gray drawing). Now, assuming that $\delta_\text{u}=0$, the left and right lines have the same length, so that the wing tips are forced to leave the neutral
configuration (see Fig. \ref{fig:PhiThetaInputDist}, black drawing), in the same way as a steering input were acting on the wing. In other words, in order for the wing to be in a neutral configuration in the presence of a value of $\phi\neq0$, the two steering lines should have a difference of length equal to $d\, \sin{(\phi)}\cos{(\theta)}$. Hence, we can compute the geometric input as:
\begin{equation}
 \delta_\text{g}(t) = -d\, \sin{(\phi(t))}\cos{(\theta(t))}.\label{eqn:PhiThetaInputDist}
\end{equation}
We note that the geometric input is always bounded by $d$, moreover with an increasing $\theta$ position of the wing its value decreases and becomes
eventually zero if the wing is at the zenith position of the wind window.
\begin{figure}[hbt]
\centerline{
\includegraphics[bbllx=92mm,bblly=133mm,bburx=192mm,bbury=210mm,width=8cm,clip]{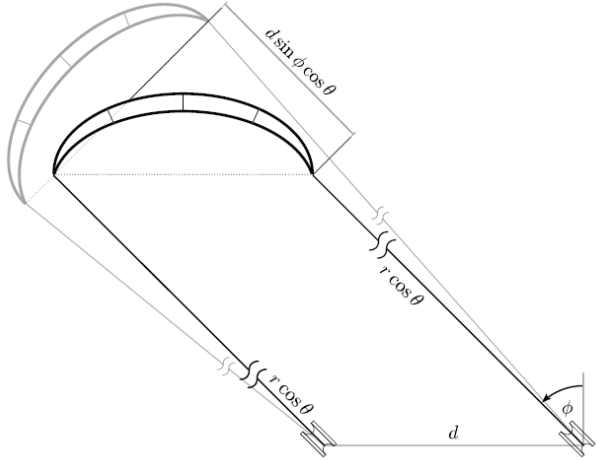}
} \caption{Sketch of the wing in a generic $(\theta,\phi)$ position as seen from above, and related geometric input $\delta_\text{g}=-d\, \sin{(\phi)}\, \cos{(\theta)}$.}\label{fig:PhiThetaInputDist}
\end{figure}
The minus sign in \eqref{eqn:PhiThetaInputDist} derives from the fact that the geometric  input imposes a counter-clockwise turn (as seen from the GU) for $\phi < 0$ (i.e. $\dot{\gamma}>0$) and a clockwise turn for $\phi > 0$ (i.e. $\dot{\gamma}<0$).

\subsection{Control-oriented model for tethered wings}\label{S:control_model}
In \cite{ErSt12},  a simple model was presented and used for the control design, where the time derivative of the heading angle of the wing is given as a function of the control variable. Here, we consider a similar simplified model, where we use the velocity angle instead of the heading angle. This choice is supported by the assumption of small sideslip angle, as formally stated below:
\begin{assumption}
The difference between the velocity angle $\gamma(t)$ and the heading angle $\xi(t)$ is negligible, i.e. the effective wind projected onto the tangent plane to the wind window at the wing's location is equal to the wing's velocity $\vec{v}(t)$. Moreover, all the forces in the direction of vector $\vec{v}(t)$ are negligible as compared to lift and drag.\hfill$\blacksquare$
 \label{as:Crosswind}
\end{assumption}
Assumption \ref{as:Crosswind} is common in the analysis of airborne wind energy generators (see, e.g., \cite{Loyd80,FaMP11}) and it is reasonable whenever the wing is flying downwind roughly perpendicularly to the wind flow. As a further justification for this assumption, we show in Fig. \ref{fig:GammaVsBeta} the good matching between the values of $\gamma(t)$ and $\xi(t)$ measured during figure-eight paths carried out with our small-scale prototype and a 9-m$^2$ power kite.
 \begin{figure}[hbt]
 \centerline{
 \includegraphics[bbllx=10mm,bblly=68mm,bburx=195mm,bbury=204mm,width=7.5cm,clip]{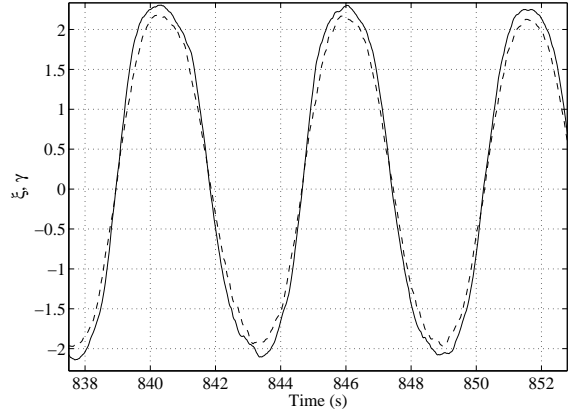}
 } \caption{Experimental results. Angles $\gamma$ (solid line) and $\xi$ (dotted) measured during controlled flights in crosswind conditions. Wing size: 9-m$^2$.}\label{fig:GammaVsBeta}
 \end{figure}
With Assumption \ref{as:Crosswind} in mind, the model proposed in \cite{ErSt12} can be written in the following form:
\begin{equation}\label{eqn:gamma_model_form}
\dot{\gamma}(t) \simeq \tilde{K}(t)\delta(t) + \tilde{T}(t).
\end{equation}
In particular, in \cite{ErSt12} the term $\tilde{K}(t)$ is given by a constant gain, $g_K$, multiplied by the magnitude of the effective wind speed aligned with the wing's heading. Such a model was justified mainly through experimental
results in \cite{ErSt12}, where the gain $g_K$ was derived empirically. In \cite{BaOc12}, a model of the form \eqref{eqn:gamma_model_form} is also used, and justified by some assumptions. Equation \eqref{eqn:gamma_model_form} provides indeed a model well-suited for control design, with the right balance between accuracy and simplicity, however neither \cite{ErSt12} nor \cite{BaOc12} derived an explicit relationship between the model's parameters and the main characteristics of the system, such as wing size, mass, or efficiency. The theoretical result we are presenting next is aimed to fill this gap, by linking equation \eqref{eqn:gamma_model_form} to the first-principle model recalled in section \ref{SS:model}. Moreover, we present experimental data that confirm the validity of our result. We consider the following assumption on the roll angle $\psi(t)$:
\begin{assumption}
 The angle $\psi(t)$ is sufficiently small to linearize its trigonometric functions. \hfill$\blacksquare$
 \label{as:SmallInput}
\end{assumption}
Assumption \ref{as:SmallInput} is reasonable in the considered context, since for example the prototype at the UC Santa Barbara operates with $\psi\simeq\pm 7.5^\circ$ (for a 12-m$^2$ wing with wingspan $d_s=3.1\,$m) up to $\psi\simeq\pm 12.8^\circ$ (for a 6-m$^2$ wing with wingspan $d_s=1.8\,$m).\\
We can now state our theoretical result.
\begin{proposition}\label{P:gamma_dot} Let assumptions \ref{as:Crosswind}-\ref{as:SmallInput} hold. Then, equation \eqref{eqn:gamma_model_form} holds with:
\begin{subequations}\label{eqn:GammaCtrlLaw}
\begin{align}
 \tilde{K}(t) =
 \frac{\rho\,C_L(t)\,A}{2md_s}\left(1+\frac{1}{E_{eq}^2(t)}\right)^2\,|\vec{v}(t)|\label{eqn:GammaCtrlLaw1}\\
 \tilde{T}(t)=\frac{g\cos{(\theta(t))}\, \sin{(\gamma(t))}}{|\vec{v}(t)|}+\sin{(\theta(t))}\, \dot{\phi}(t),\label{eqn:GammaCtrlLaw2}
\end{align}
\end{subequations}
where $|\vec{v}(t)|$ is the magnitude of the wing's velocity, and $E_{eq}(t)\doteq C_L(t)/C_{D,eq}(t)$.
\end{proposition}
\begin{proof}See the Appendix.\end{proof}
Proposition \ref{P:gamma_dot} provides an explicit link between the main lumped parameters of the wing, like area, mass and lift coefficient, and the gain and external disturbance of \eqref{eqn:gamma_model_form}. It is worth elaborating more on this result and its implications. According to equation
\begin{figure*}[hbt]
\centerline{
\begin{tabular}{cc}
(a) & (b)\\
\includegraphics[bbllx=17mm,bblly=58mm,bburx=200mm,bbury=203mm,width=7cm,clip]{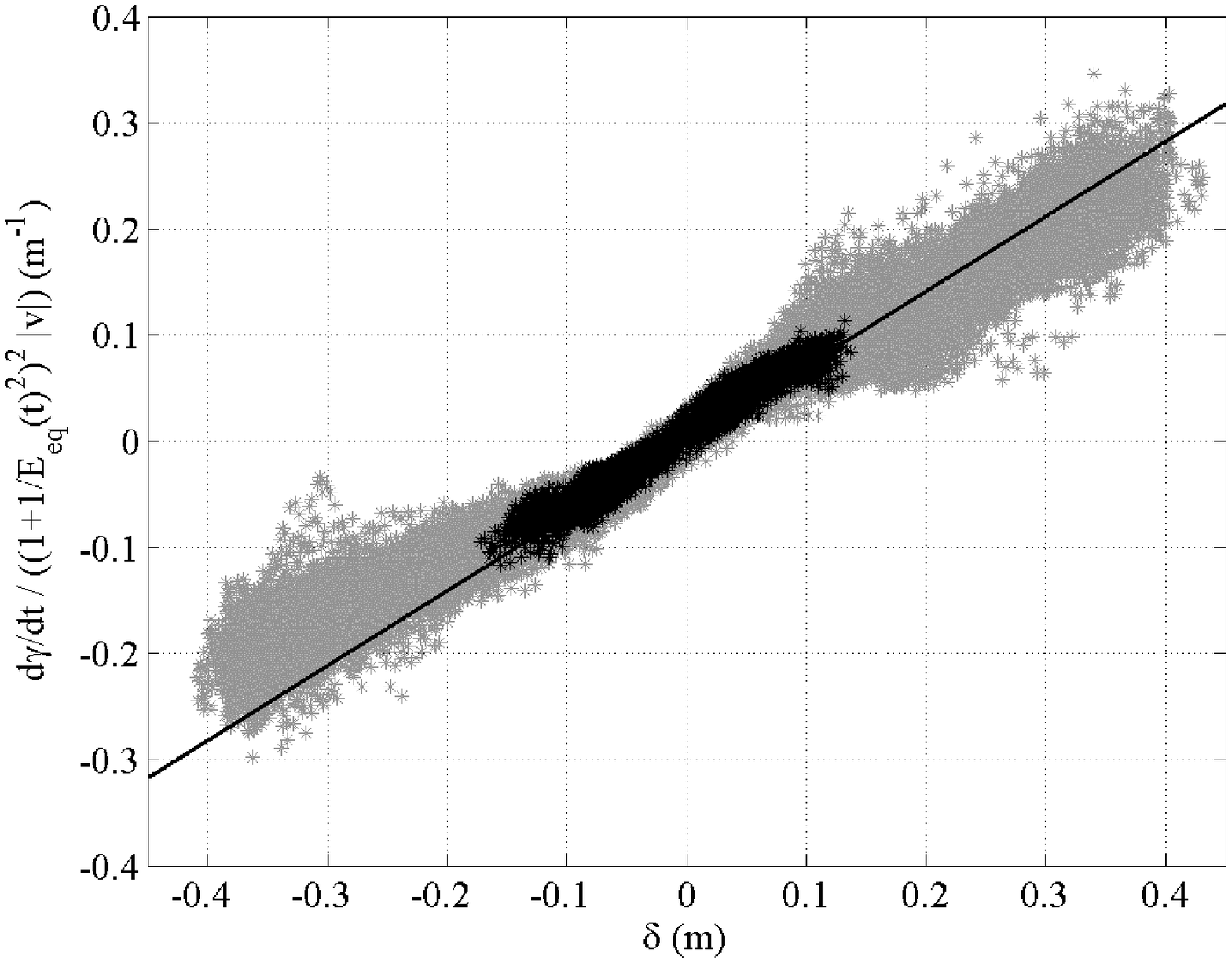}&
\includegraphics[bbllx=17mm,bblly=58mm,bburx=200mm,bbury=203mm,width=7cm,clip]{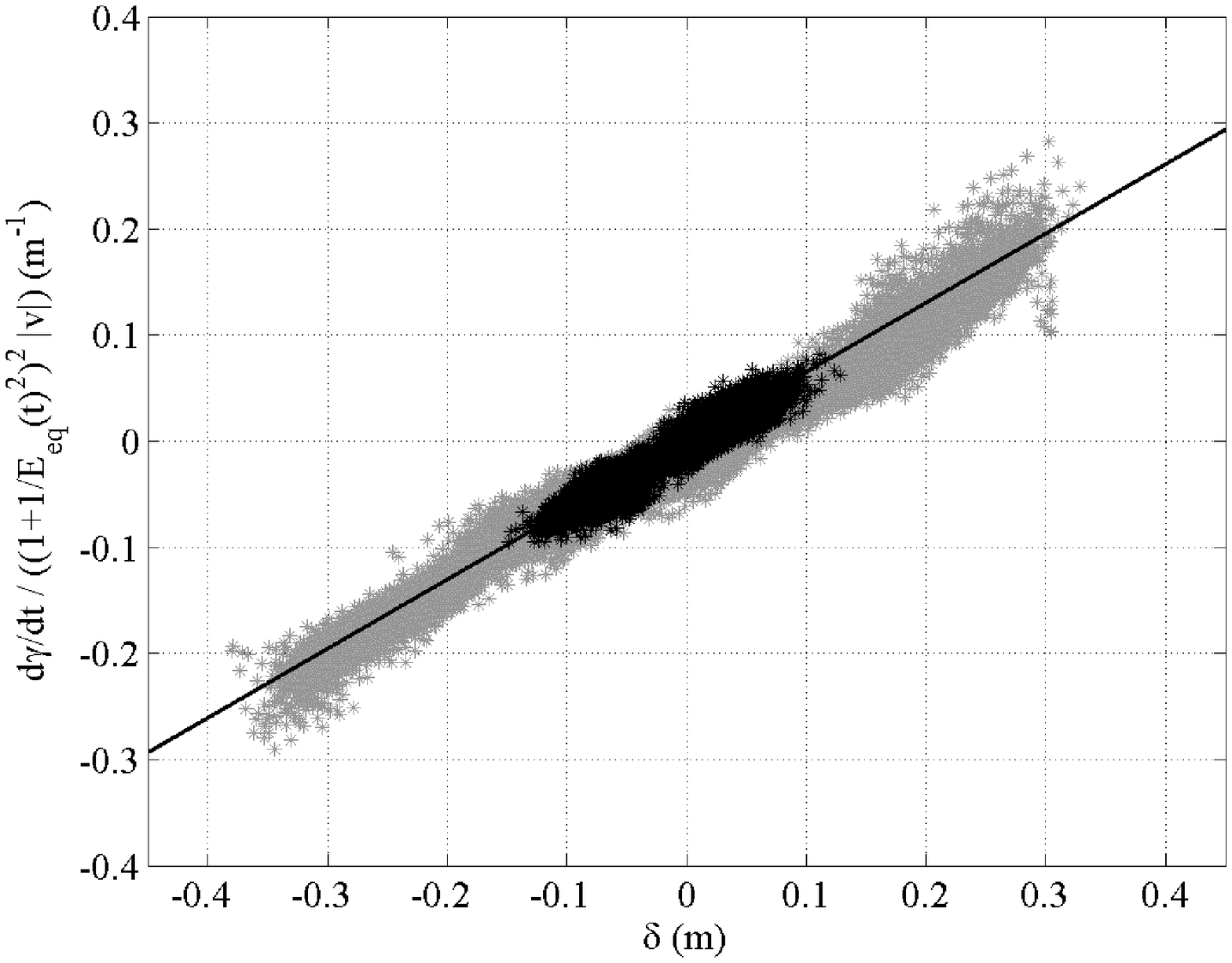}\\
(c)&(d)\\
\includegraphics[bbllx=13mm,bblly=58mm,bburx=200mm,bbury=203mm,width=7cm,clip]{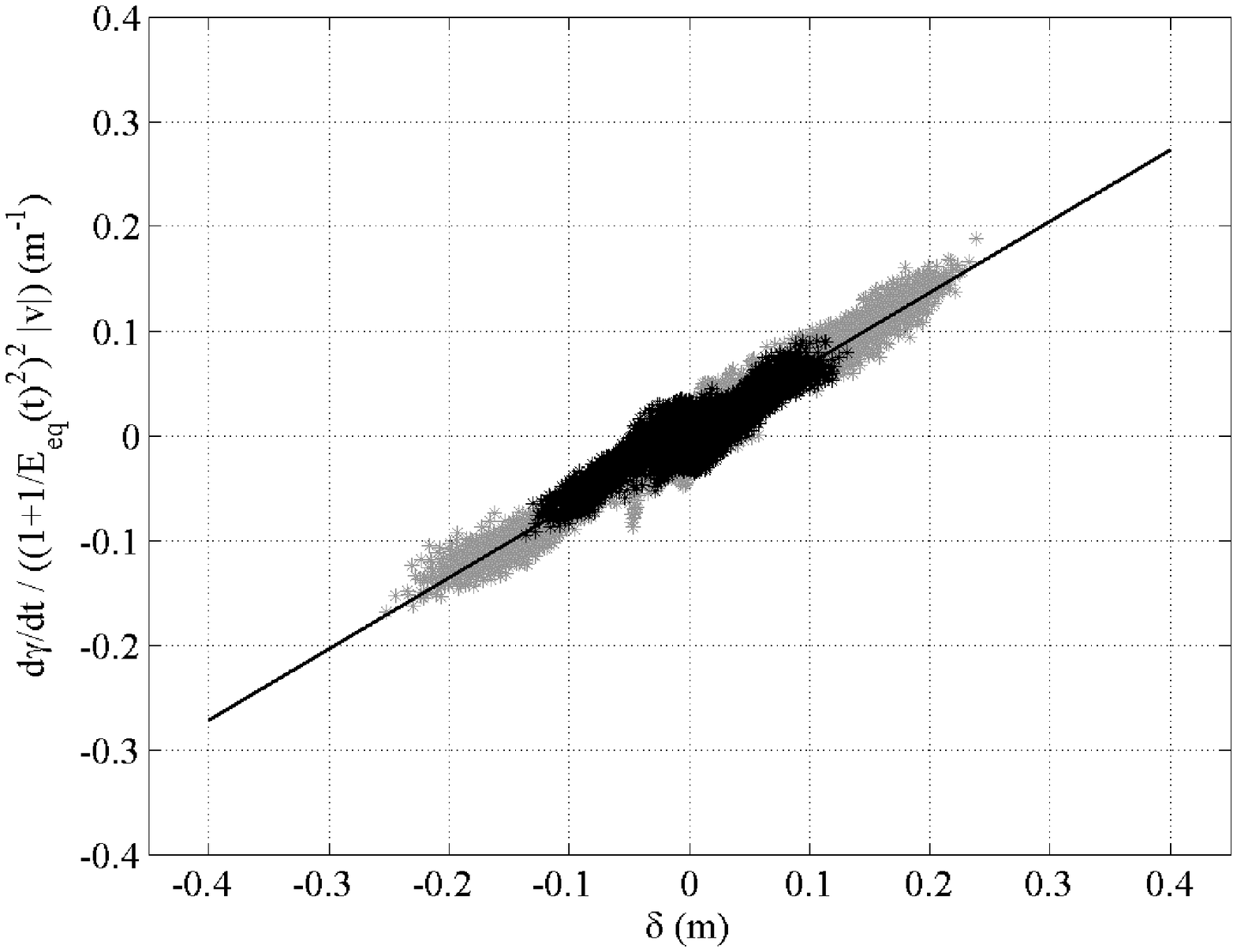}&
\includegraphics[bbllx=14mm,bblly=70mm,bburx=194mm,bbury=202mm,width=7cm,clip]{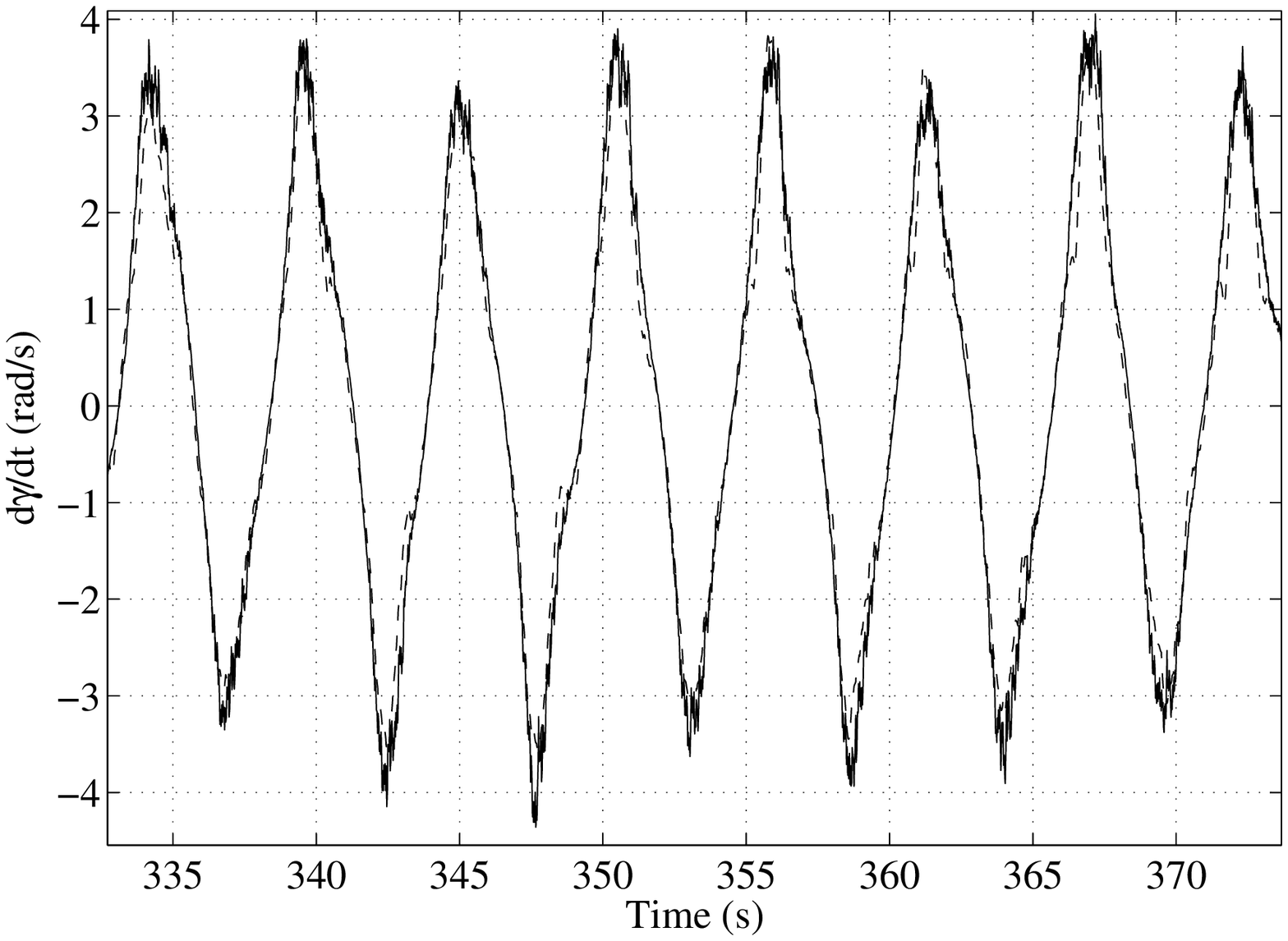}\\
\end{tabular}} \caption{Experimental results. (a)-(c): Comparison between the measured values of $\dot{\gamma}(t)/\left(\left(1+\frac{1}{E_{eq}^2(t)}\right)^2\,|\vec{v}(t)|\right)$ as a function of the steering input $\delta(t)$ (gray and black dots) and the theoretical linear relationship given by the gain $\frac{\rho\,C_L\,A}{2md_s}$ as per equation \eqref{eqn:GammaCtrlLaw1} (solid line).  The gray dots represent experimental data collected in the whole range of $\theta,\,\phi$ spanned by the wing during operation, while the black dots represent values collected when $|\phi|\leq5^\circ$, i.e. in crosswind conditions. (a) Airush One$^\circledR$ 6 kite, (b) Airush One$^\circledR$ 9 kite, (a) Airush One$^\circledR$ 12 kite. The lumped parameters for the kites are reported in Table \ref{T:wing_param}.
(d): Comparison between the value of $\dot{\gamma}$ obtained in experimental tests (solid) and the one estimated using the simplified model and the result of Proposition \ref{P:gamma_dot} (dashed). Wing size: 6-m$^2$.} \label{F:power_vs_parameters}
\end{figure*}
\eqref{eqn:gamma_model_form}, there is basically an integrator between the control input $\delta(t)$ and the velocity angle $\gamma(t)$, with a time-varying gain. In the first term of \eqref{eqn:GammaCtrlLaw1}, we can see that such a gain increases as the wing's speed does; thus, a larger speed provides higher control authority but it can also bring forth stability issues, if the control system is not properly designed. For a given wing flying in crosswind conditions, the term $\frac{\rho\,C_L(t)\,A}{2md_s}\left(1+\frac{1}{E_{eq}^2(t)}\right)^2$ does not change significantly during operation, hence supporting the model proposed by \cite{ErSt12}, where, as mentioned above, a constant gain $g_K$, multiplied by the effective wind speed, is used to compute $\tilde{K}$. Equation \eqref{eqn:GammaCtrlLaw1} also implies that a larger area-to-mass ratio gives in general a higher gain $\tilde{K}$, and that the steering behaviors of wings with similar design (i.e. similar aerodynamic coefficients) but different sizes are expected to be similar, provided that the area-to-mass ratio, $A/m$, does not change much when the size is scaled up. However,  the latter consideration holds true if just the wing itself is considered, since for example the use of larger wings would require lines with larger diameter, and the consequent added mass and induced drag would generally reduce both the equivalent efficiency and the area-to-mass ratio.\\
Furthermore, since the wing's speed is roughly proportional to the effective wind speed  projected along the lines' direction (see e.g. \cite{FaMP11}), the gain $\tilde{K}$ can be also re-written as (assuming the wind is aligned with the $X$-axis, as we remarked in section \ref{SS:yo_yo_descr}):\\
\begin{equation}\label{eqn:GammaCtrlLaw3}\small
\tilde{K}(t) =
 \frac{\rho\,C_L(t)\,A\,E_{eq}(t)\cos{(\theta(t))}\cos{(\phi(t))}}{2md_s}\left(1+\frac{1}{E_{eq}^2(t)}\right)^2\,|\vec{W}|.
\end{equation}\normalsize
According to \eqref{eqn:GammaCtrlLaw3}, the value of $\tilde{K}$ is highest when the wing is flying crosswind, and it decreases as the wing approaches the borders of the wind window (i.e. $\theta\simeq0$ and/or $\phi\simeq\pm\frac{\pi}{2}$). Therefore, maneuvering the wing in these conditions requires larger control inputs and in some situations it might be not possible, leading to a complete loss of controllability.  Finally, the wing's efficiency and lift coefficient also play an important role: in particular, for fixed drag coefficient the gain is expected to grow quadratically with the lift coefficient, as it can be noted from \eqref{eqn:GammaCtrlLaw3} by considering that $E_{eq}=C_L/C_{D,eq}$.\\
The term $\tilde{T}(t)$ basically accounts for the steering effect that the gravity and apparent forces have on the wing. It can be noted that the influence of gravity gets smaller with larger wing speed, hence it can be easily dominated by the control action when the wing's lines are aligned with the wind, while it becomes more and more important as the wing moves to the side of the wind window. 
The result of Proposition \ref{P:gamma_dot} is also well-confirmed by the experimental tests we carried out with our small-scale prototype. In order to show the matching between our result and the experimental evidence, we compare the measured values of $\dot{\gamma}$, normalized by
$\left(1+\frac{1}{E_{eq}^2(t)}\right)^2\,|\vec{v}(t)|$, and the steering input $\delta(t)$: according to equation \eqref{eqn:GammaCtrlLaw1}, these two quantities are proportional through the gain given by $\frac{\rho\,C_L\,A}{2md_s}$. We show the matching between the experimental data and such a relationship in Fig.  \ref{F:power_vs_parameters}(a), (b) and (c), for a 6-m$^2$, a 9-m$^2$ and a 12-m$^2$ wing, respectively. The lumped parameters of the employed wings are resumed in Table \ref{T:wing_param}. In Fig. \ref{F:power_vs_parameters}(a)-(c), the gray dots represent experimental data collected in the whole range of $\theta,\,\phi$ spanned by the wing during operation, while the black dots represent values collected when $|\phi|\leq5^\circ$, i.e. in crosswind conditions. It can be noted that the linear relationship computed by using the lumped parameters as given by the model \eqref{eqn:GammaCtrlLaw} matches quite well with the experimental data, not only in crosswind conditions, where the underlying assumptions are valid, but also with larger values of $\phi$, in the range $\pm35^\circ$. Finally, in Fig. \ref{F:power_vs_parameters}-(d) we show an example of the matching between the time course of $\dot{\gamma}$ during figure-eight paths with the 6-m$^2$ wing and the estimate given by the model \eqref{eqn:gamma_model_form}-\eqref{eqn:GammaCtrlLaw}. Overall, the presented results provide a solid bridge between the existing dynamical models of the steering behavior of tethered wings and experimental evidence. In the next section, we introduce a new control approach for crosswind flight of tethered wings, which is based on the presented simplified model.
\begin{table}[!h]
  \centering
  \caption{Lumped parameters of the wings employed in the experimental activities.}\label{T:wing_param}
\begin{tabular}{|l|l|l|}
  \hline
  \multicolumn{3}{|l|}{Airush One$^\circledR$ 12}\highertop\\ \hline\highertop
  Area &$A$ &12$\,$m$^2$\\ \hline\highertop
  Mass &$m$ &2.9$\,$kg\\ \hline\highertop
  Wingspan &$d_s$ &3.1$\,$m\\ \hline\highertop
  Lift coefficient (average) &$C_L$ &0.85 \\ \hline\highertop
  Equivalent aerodynamic efficiency (average) &$E_{eq}$ &5.3\\ \hline\hline
   \multicolumn{3}{|l|}{Airush One$^\circledR$ 9}\highertop\\ \hline\highertop
  Area &$A$ &9$\,$m$^2$\\ \hline\highertop
  Mass &$m$ &2.45$\,$kg\\ \hline\highertop
  Wingspan &$d_s$ &2.7$\,$m\\ \hline\highertop
  Lift coefficient (average) &$C_L$ &0.8 \\ \hline\highertop
  Equivalent aerodynamic efficiency (average) &$E_{eq}$ &5.6\\ \hline\hline
   \multicolumn{3}{|l|}{Airush One$^\circledR$ 6}\highertop\\ \hline\highertop
  Area &$A$ &6$\,$m$^2$\\ \hline\highertop
  Mass &$m$ &1.7$\,$kg\\ \hline\highertop
  Wingspan &$d_s$ &1.8$\,$m\\ \hline\highertop
  Lift coefficient (average) &$C_L$ & 0.6\highertop\\ \hline\highertop
  Equivalent aerodynamic efficiency (average) &$E_{eq}$ &5.1\\ \hline
  \end{tabular}
\end{table}

\section{Control design}\label{S:control_design}
We propose a control scheme consisting of three nested loops, shown in Fig. \ref{F:control_scheme}.
The outer control loop employs the current wing position, in terms of $\theta,\,\phi$ angles, to compute a reference velocity angle, $\gamma_{\text{ref}}$, for the middle control loop.
\begin{figure*}[!phtb]
 \begin{center}
  \includegraphics[bbllx=9mm,bblly=203mm,bburx=273mm,bbury=257mm,width=18cm,clip]{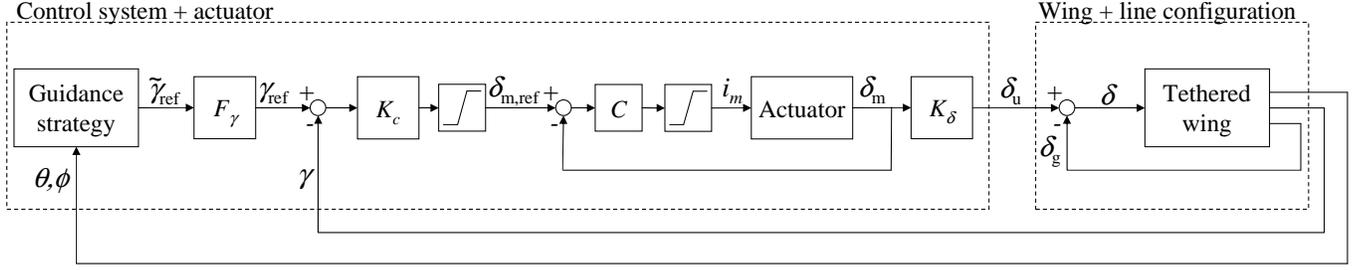}
  \caption{Overview of the proposed control system.}
  \label{F:control_scheme}
 \end{center}
\end{figure*}
The latter employs $\gamma$ as feedback variable and it has the objective of tracking $\gamma_{\text{ref}}$, by setting a suitable position reference, $\delta_{\text{m,ref}}$, for the actuator installed on the GU. The innermost control loop then employs a feedback of the motor position $\delta_\text{m}$ to command the motor's current $i_\text{m}$, in order to track the desired position $\delta_{\text{m,ref}}$. This control structure allows to separate the nonlinear part of the controller, which is all kept at the outermost level, from the linear one, hence obtaining two simple controllers for the middle and inner loops, for which we can carry out a theoretical robustness analysis. The control algorithm for the outer loop is given by a quite simple guidance strategy, hence making the whole control system very suitable for implementation and experimental testing. In the next sections, we describe each control loop in details. We first adopt a general notation for the involved design parameters and we then provide in section \ref{S:results} the specific numerical values employed in our tests.

\subsection{Position control loop}\label{SS:innermost}
The innermost loop consists of a standard position control system with an electrical DC brushed motor and a linear motion system, based on a lead screw mechanism. By neglecting high-order effects, in the absence of external disturbances the dynamics of the actuator can be modeled as:
\begin{equation}\label{eqn:motor_dyn}
\dfrac{\ddot{\delta}_\text{m}(t)}{\omega_\text{m}}=-\dot{\delta}_\text{m}(t)+K_\text{m}i_\text{m}(t),
\end{equation}
where $\delta_\text{m}(t)$ is the actuator's position in m, $K_\text{m},\,\omega_\text{m}$ are parameters depending on the actuator's characteristics and $i_\text{m}(t)$ is the commanded current in A. Moreover, the prototype is equipped with a series of pulleys such that, for a given value of $\delta_\text{m}$, the corresponding difference of length of the steering lines is equal to
\begin{equation}\label{eqn:mot_diff}
\delta_\text{u}(t)=K_\delta\,\delta_\text{m}(t),
\end{equation}
with $K_\delta$ being a constant gain. The value of $\delta_\text{m}$ is measured with high accuracy by an optical rotary incremental encoder, suitably scaled to obtain the linear position of the actuator from the motor's angular position. We use standard loop-shaping control design techniques (see e.g. \cite{SkPo05}) to design a cascade feedback controller $C$ for this control level. The control scheme is shown in Fig. \ref{F:control_scheme_innermost}.
\begin{figure}[!phtb]
 \begin{center}
  \includegraphics[bbllx=7mm,bblly=148mm,bburx=270mm,bbury=196mm,width=8.5cm,clip]{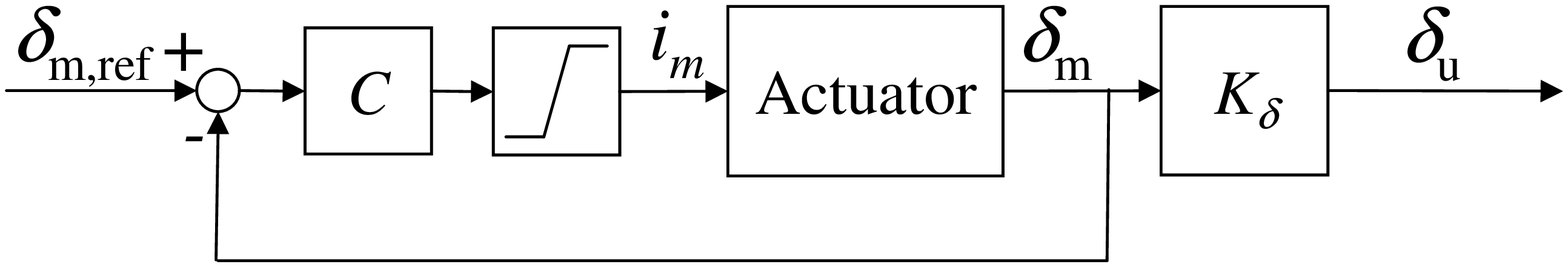}
  \caption{Scheme of the actuator position control loop.}
  \label{F:control_scheme_innermost}
 \end{center}
\end{figure}
We note that the actuator has  current limits of $\pm\overline{i}_\text{m}$, included as a saturation in Fig. \ref{F:control_scheme_innermost}. Such a saturation does not give rise to integrator windup problems, since the employed controller $C$ does not have integral action. Yet, the presence of an integrator in the plant \eqref{eqn:motor_dyn} still yields a closed-loop transfer function with unitary gain between the reference and the output. In particular, the closed loop system for the inner loop results to be of the form:
\begin{equation}\label{eqn:inner_loop_cl}
\ddot{\delta}_\text{m}(t)+2\zeta_\text{cl}\omega_\text{cl}\dot{\delta}_\text{m}(t)+\omega_\text{cl}^2\delta_\text{m}(t)=
\omega_\text{cl}^2\delta_\text{m,ref}(t),
\end{equation}
where $\delta_\text{m,ref}(t)$ is the reference position provided by the controller for the wing's velocity angle, which is described next.

\subsection{Velocity angle control loop}\label{SS:inner}
The control design for the middle loop exploits the control-oriented model \eqref{eqn:GammaCtrlLaw} derived in section \ref{S:control_model}. The control input is the actuator position $\delta_\text{m,ref}(t)$  (i.e. the reference for the innermost control loop described in section \ref{SS:innermost}), and the feedback variable is the wing's velocity angle $\gamma(t)$. The controller is given by a simple proportional law :
\begin{equation}\label{eqn:middle_loop_ctrl}
\delta_\text{m,ref}(t)=K_c\left(\gamma_\text{ref}(t)-\gamma(t)\right),
\end{equation}
where $\gamma_\text{ref}(t)$ is the target velocity angle provided by the outer control loop, and $K_c$ is a scalar gain to be chosen by the designer. The value of $K_\text{c}$ is the only design parameter for the velocity angle controller, and it can be tuned at first by using the  equations \eqref{eqn:gamma_model_form}-\eqref{eqn:GammaCtrlLaw} with some nominal system
parameters, and then via experiments. Moreover we note that, with the proposed approach, $K_c$ can be tuned in order to robustly stabilize the control loop. In particular, for a fixed value of $K_c$, by exploiting the model \eqref{eqn:GammaCtrlLaw} we can carry out a robustness analysis of the control system comprising the velocity angle loop and the actuator position loop. Such analysis allows us to discern for what range of wind speeds and wing's parameters the considered static gain $K_c$ is able to stabilize the velocity angle control system. In order to do so, we re-write the dynamical system given by equations \eqref{eqn:GammaCtrlLaw}-\eqref{eqn:mot_diff} and \eqref{eqn:inner_loop_cl}-\eqref{eqn:middle_loop_ctrl} in terms of velocity angle tracking error, $\Delta_\gamma(t)\doteq\gamma_\text{ref}-\gamma(t)$:\\
\begin{equation}\label{eqn:middle_loop_cl}\small
\left[
\begin{array}{c}
\dot{\Delta}_\gamma(t)\\
\dot{\delta}_\text{m}(t)\\
\ddot{\delta}_\text{m}(t)
\end{array}\right]=\underbrace{\left[
\begin{array}{ccc}
0&-\tilde{K}K_\delta&0\\
0&0&1\\
K_c\omega_\text{cl}^2&-\omega_\text{cl}^2&-2\zeta\omega_\text{cl}
\end{array}
\right]}\limits_{A_\text{cl}(\tilde{K})}
\left[
\begin{array}{c}
\Delta_\gamma(t)\\
\delta_\text{m}(t)\\
\dot{\delta}_\text{m}(t)
\end{array}\right]+w(t).
\end{equation}\normalsize
In \eqref{eqn:middle_loop_cl}, the term $\tilde{K}$ corresponds to the gain in \eqref{eqn:GammaCtrlLaw} and depends on the system's parameters as well as the wind and flight conditions. The term $w(t)$ accounts for the effects of gravity and apparent forces of \eqref{eqn:GammaCtrlLaw}, as well as the forces exerted by the lines on the innermost control system, and it can be seen as a bounded disturbance acting on the system. System \eqref{eqn:middle_loop_cl} has time-varying, uncertain linear dynamics characterized by the matrix $A_\text{cl}(\tilde{K})$. In fact, the scalar $\tilde{K}(|\vec{v}(t)|,A,E_{eq}(t),C_L(t),m,d_s)$ is a function of several parameters, like the wing's speed magnitude and the wing's efficiency and lift coefficient, that are not precisely known and vary over time, since for example the aerodynamic coefficients depend on the wing's angle of attack, which changes in time according to the flight conditions. However, upper and lower bounds for all of the involved parameters can be easily derived
on the basis of the available knowledge on the system, and these bounds can be employed to compute limits $\tilde{K}^i,\,i=1,\,2,$ such that $\tilde{K}\in[\tilde{K}^1,\,\tilde{K}^2]$. Then, any possible matrix $A_\text{cl}(\tilde{K})$ results to be contained in the convex hull defined by the vertices  $A_\text{cl}(\tilde{K}^1),\,A_\text{cl}(\tilde{K}^2)$. After computing these vertices, existing results based on quadratic stability and linear matrix inequalities (LMI) can be used to assess the robust stability of system \eqref{eqn:middle_loop_cl}. In particular, system \eqref{eqn:middle_loop_cl} results to be robustly stable if there exists a positive definite matrix $P=P^T\in\mathbb{R}^{3\times3}$ such that (see e.g. \cite{Amato06}):
\begin{equation}\label{eqn:stability_cond}
A_\text{cl}^T(\tilde{K}^i)P+PA_\text{cl}(\tilde{K}^i)\prec0,\,i=1,\,2,
\end{equation}
where $^T$ stands for the matrix transpose operation. Condition \eqref{eqn:stability_cond} can be easily checked by using an LMI solver. In section \ref{S:results} we show that indeed a unique value of $K_\text{c}$ guarantees robust stability of the velocity angle control loop for a wide range of operating conditions. Since the wing's velocity can be measured quite accurately, one could also adopt a gain scheduling approach and make the gain in \eqref{eqn:middle_loop_ctrl} depend on $|\vec{v}(t)|$, in order to improve the performance. However, in our experimental tests this was not needed, as the control systems achieved satisfactory performance for all the conditions of wind speed (and, consequently, of wing speeds) that we experienced, as well as with all the three different wings we tested.\\
Finally we remark that, due to the physical limitations of the actuator, the position $\delta_\text{m}(t)$ is constrained in the range $\pm\overline{\delta}_\text{m}$, where $\overline{\delta}_\text{m}$ is a positive scalar. Similarly, we saturate the reference position $\delta_\text{m,ref}(t)$ in the same range. However, we note that these saturations are never active during operation (we show this fact with measured data in the next section), hence the stability analysis reported above is still valid even if it does not account explicitly for the input limits. Fig. \ref{F:middle_loop} shows a scheme of the velocity angle controller.
\begin{figure}[!phtb]
 \begin{center}
  \includegraphics[bbllx=8mm,bblly=118mm,bburx=273mm,bbury=203mm,width=6cm,clip]{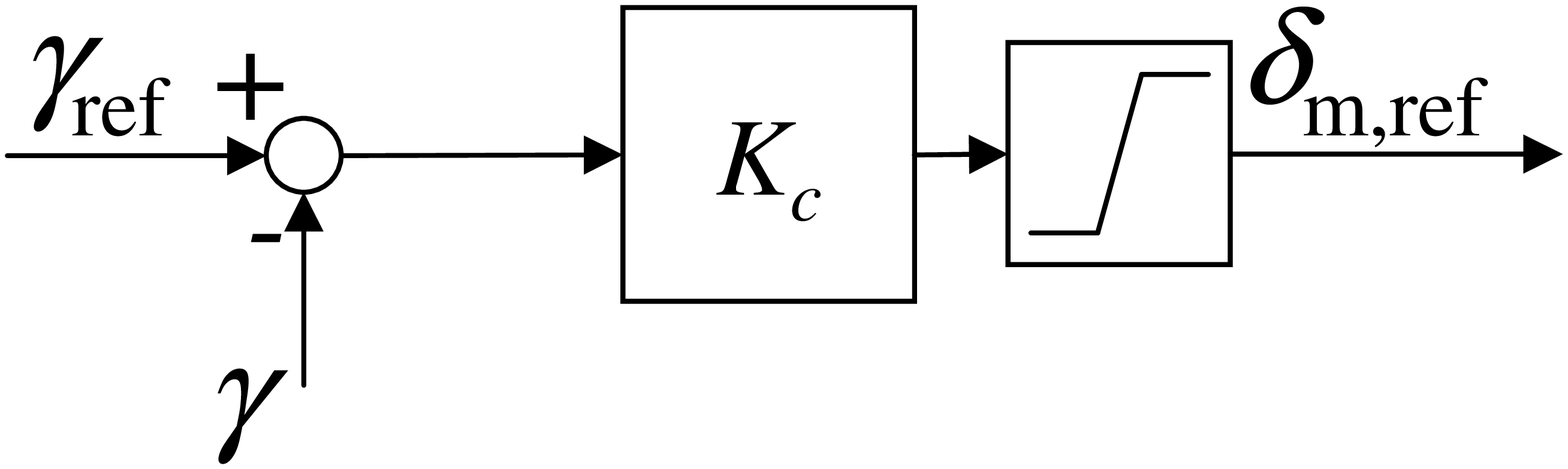}
  \caption{Scheme of the velocity angle controller.}
  \label{F:middle_loop}
 \end{center}
\end{figure}

\subsection{Outer control loop}\label{SS:outer}
The outer control loop (shown in Fig. \ref{F:outer_loop}) is responsible for providing the velocity angle control loop with a reference heading $\gamma_\text{ref}(t)$.
\begin{figure}[!phtb]
 \begin{center}
  \includegraphics[bbllx=38mm,bblly=182mm,bburx=123mm,bbury=237mm,width=5.5cm,clip]{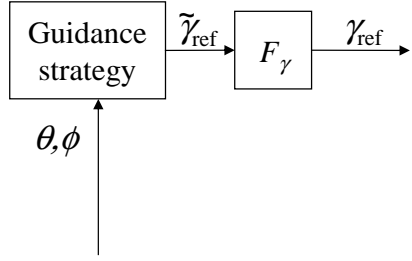}
  \caption{Scheme of the guidance strategy.}
  \label{F:outer_loop}
 \end{center}
\end{figure}
The goal is to have a control algorithm able to make the wing fly along figure-eight paths, with few tuning parameters. Hence, we want to avoid the need to pre-compute a whole trajectory to be used as reference, as it has been done in previous works \cite{IlHD07,BaOc12}. In fact, pre-computed reference paths are based on some mathematical model of the system, with the consequent unavoidable issues of model mismatch and approximation that might give rise to problems related to stability and attractiveness of the chosen trajectory. In \cite{ErSt12}, a  bang-bang like strategy to set the reference heading of the wing, avoiding the use of pre-computed reference flying paths, has been described. Here, we propose a similar approach, whose advantage is to provide a simple and explicit link between the tuning parameters and the position of the resulting paths in the wind window.\\
In our approach we define two fixed reference points in the ($\phi,\,\theta$) plane, denoted by $P_- = (\phi_-,\theta_-)$ and
$P_+ = (\phi_+,\theta_+)$, with $\phi_-<\phi_+$ (see Fig. \ref{fig:TargetPoints}).
\begin{figure}[htb]
\centerline{
  \includegraphics[bbllx=84mm,bblly=142mm,bburx=195mm,bbury=200mm,width=9cm,clip]{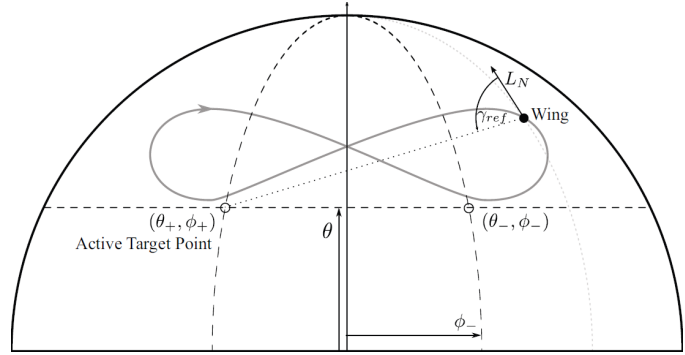}}
  \caption{Sketch of the control strategy for the outermost loop. Wind window projected on the $(Y,Z)$ plane (black solid line), target points $P_-,\,P_+$ (`$\circ$'), traces of points with constant $\theta$ and $\phi$ (dashed lines), example of wing's path (gray solid line), and example of how the reference $\gamma_\text{ref}$ is computed for a given wing position (`$\bullet$').}
\label{fig:TargetPoints}
\end{figure}
The controller computes a new value of the reference velocity angle at discrete time instants. At each time step $k\in\mathbb{Z}$, one of the two reference points is set as the active target $P_\text{a}(k)=(\phi_\text{a}(k),\,\theta_\text{a}(k))$, according to a switching strategy that we describe next. Then, a first target velocity angle $\tilde{\gamma}_\text{ref}(k)$ is computed on the basis of the measured values of $\theta(k)$ and
$\phi(k)$ as follows:
\begin{equation}
 \tilde{\gamma}_\text{ref}(k) = \arctan{\left(\frac{(\phi_\text{a}(k)-\phi(k))\cos{(\theta(k))}}{\theta_\text{a}(k) - \theta(k)}\right)},\label{eqn:GammaRef}
\end{equation}
i.e. in order to make the wing's velocity vector point towards the active target (compare the definition of $\gamma(t)$ in \eqref{eqn:Gamma}). Finally, the actual reference $\gamma_\text{ref}$ is computed as the output of a 2$^\text{nd}$-order Butterworth filter $F_\gamma$, whose input is $\tilde{\gamma}_\text{ref}$, with cutoff frequency $\omega_\gamma$, to be tuned by the control designer, in order to provide a smooth reference as input to the velocity angle control loop.
The target points are switched according to the following strategy:
\begin{equation}\label{eqn:switch}
\begin{array}{l}
 \text{If $\phi(k)<\phi_-$ then $P_\text{a}(k)=P_+$}\\ \text{If $\phi(k)>\phi_+$ then $P_\text{a}(k)=P_-$}\\
 \text{Else $P_\text{a}(k)=P_\text{a}(k-1)$}.
 \end{array}
\end{equation}
Thus, the target point is switched when the measured value of $\phi$ is outside the interval $[\phi_-,\,\phi_+]$. By how we defined the velocity angle $\gamma$, after the target point has been switched the wing will start turning, under the action of the inner control loops, following ``up-loops'', i.e. pointing first towards the zenith of the wind window and then to the other target point.
Equations \eqref{eqn:GammaRef}-\eqref{eqn:switch} provide the core of the controller we propose for setting the reference velocity angle. We note that also this control loop involves few parameters, i.e. the target points $P_-,\,P_+$, which can be intuitively tuned by using simplified equations of crosswind motion like those presented in \cite{FaMP11}, and the cutoff frequency $\omega_\gamma$, which can be tuned in order to have sharper (for higher $\omega_\gamma$) or wider (for smaller $\omega_\gamma$) turns of the wing (we provide an example obtained from our experiments in the next section). Finally, we note that the proposed control strategy does not rely on any measurement or estimate of the wind speed, however the interval $[\phi_-,\,\phi_+]$ shall
be centered around the wind direction, in order to make the wing fly in crosswind conditions. A rough measure of the wind direction can be obtained by means of standard wind vanes, moreover an estimate can be obtained from other measured quantities, like the wing's speed and line forces.

\section{Experimental results}\label{S:results}
We implemented the controller described in section \ref{S:control_design} on a real-time machine made by SpeedGoat$^\circledR$ and programmed
with the xPC Target$^\circledR$ toolbox for MatLab$^\circledR$. We employed three different wings in our tests: a 6$\,$m$^2$, a 9$\,$m$^2$ and a 12$\,$m$^2$ Airush$^\circledR$ One power kites. The lumped parameters of the wings are reported in Table \ref{T:wing_param}. The sampling frequencies we used for the control loops are 100 Hz for the innermost controller and 50 Hz for the middle and outermost ones. A movie of the experimental tests is available online \cite{Wing_movie}.
Table \ref{T:system_params} shows the main system's and control parameters, in addition to those reported in Table \ref{T:wing_param}.
\begin{table}[!h]
  \centering
  \caption{System and control parameters.}\label{T:system_params}
\begin{tabular}{|l|l|l|}
  \hline
  \multicolumn{3}{|l|}{System's parameters}\highertop\\ \hline\highertop
  Actuator gain &$K_\text{m}$& $0.73\,$m/(s\,A)\\ \hline\highertop
  Actuator pole & $\omega_\text{m}$ & $1.9\,$rad/s\\ \hline\highertop
    Position limits &$\overline{\delta}_\text{m}$& 0.35 m\\ \hline\highertop
      Current limits &$\overline{i}_\text{m}$& 10 A\\ \hline\highertop
  Distance between & &\\
  steering lines' attachment points &$d$ &0.5$\,$m\\ \hline\highertop
  Tether length  &$r$ &30$\,$m\\ \hline\highertop
  Gain between motor position&&\\
  and line lengths' difference& $\delta_\text{u}$&4\\ \hline\highertop
  Air density &$\rho$ &1.2$\,$kg/m$^3$\\ \hline\hline
  \multicolumn{3}{|l|}{Position control loop}\highertop\\ \hline\highertop
  Sampling frequency & &100 Hz \\ \hline\highertop
  Damping &$\zeta_\text{cl}$ &0.7 \\ \hline\highertop
  Natural frequency  &$\omega_\text{cl}$ &78$\,$rad/s\\ \hline
  \hline
  \multicolumn{3}{|l|}{Velocity angle control loop}\highertop\\ \hline\highertop
  Sampling frequency & &50 Hz \\ \hline\highertop
  Feedback gain&$K_c$ &0.046 m/rad\\ \hline
  \hline
  \multicolumn{3}{|l|}{Guidance strategy}\highertop\\ \hline\highertop
  Sampling frequency & &50 Hz \\ \hline\highertop
    Target points&$P_+,P_-$ &variable \\ \hline\highertop
  Butterworth filter&&\\
  cutoff frequency&$\omega_\gamma$ &variable \\ \hline
  \end{tabular}
\end{table}
In Table \ref{T:system_params}, the values of the target points  $P_+,\,P_-$ and of the cutoff frequency $\omega_\gamma$ are indicated as ``variable'' because we tested different values, in order to assess the influence of these tuning parameters on the obtained performance. The first set of results that we show (Figs. \ref{F:innermost_loop_perf}-\ref{fig:10LoopOverlay3D}) is related to the 9-m$^2$ wing; the target points' coordinates are  $\theta_+=\theta_-=0.35\,$rad, $\phi_-=-0.2\,$rad, $\phi_+=0.2\,$rad and the value of $\omega_\gamma$ is 0.25 Hz.\\
Fig. \ref{F:innermost_loop_perf}
\begin{figure}[!htb]
\begin{center}
  \includegraphics[bbllx=11mm,bblly=64mm,bburx=189mm,bbury=209mm,width=8cm,clip]{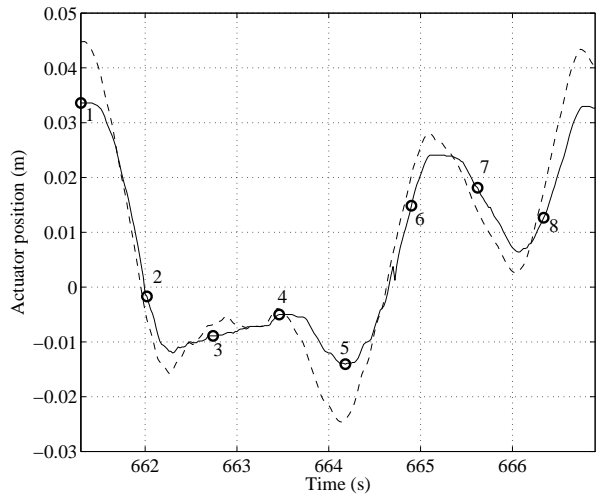}
  \caption{Experimental results. Example of  courses of target position $\delta_\text{m,ref}(t)$ (dashed line) and actual position
           $\delta_\text{m}(t)$ (solid line) of the actuator obtained with the designed innermost control loop,
           during automatic flight tests of a 9-m$^2$ wing. The numbered circles corresponds to the conditions highlighted in Fig. \ref{fig:PathForceGamma}.  Employed wing: 9 m$^2$. Guidance parameters: $\theta_+=\theta_-=0.35\,$rad, $\phi_-=-0.2\,$rad, $\phi_+=0.2\,$rad; $\omega_\gamma=0.25\,$Hz.}
\label{F:innermost_loop_perf}
\end{center}
\end{figure}
shows an example of the courses of $\delta_\text{m,ref}(t)$ and $\delta_\text{m}(t)$ obtained during the experiments with the employed
controller for the innermost control loop. Note that, for the sake of simplicity,  we did not include the presence of disturbances in the
description of this control loop in section \ref{SS:innermost}. Indeed the forces applied by the wing's lines on the actuator can be modeled
as an additive disturbance at this level.  The effect of such  disturbance can be clearly seen between points 8-1 and around point 5 in Fig. \ref{F:innermost_loop_perf}, where there is some error between the reference and actual position. By how the machine has been designed, the force exerted by each steering line on the actuator is equal to twice the force on the line. This amounts to approximately 800 N at the mentioned points in the figure. The tracking error induced by such disturbance does not influence much the performance of the overall system. If needed, the position controller can be tuned to achieve better tracking performance, at the cost of higher energy consumption. The flown path corresponding to  Fig. \ref{F:innermost_loop_perf} can be seen in Fig. \ref{fig:PathForceGamma},
\begin{figure}[htb]
 \begin{center}
  \includegraphics[bbllx=9mm,bblly=53mm,bburx=191mm,bbury=280mm,width=8cm,clip]{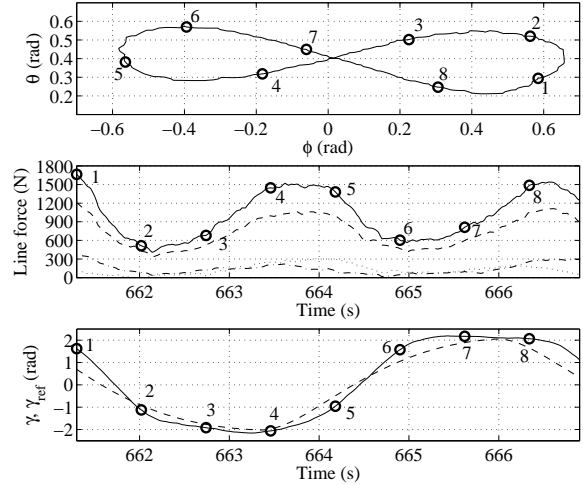}
  \caption{Experimental results. Single figure-eight path obtained during automatic test flights with about $2.4\,$m/s wind speed. From top to bottom: flying path in $(\phi,\theta)$ coordinates, course of the total force acting on the lines (solid line) and of the forces acting on the left (dotted), right (dash-dot) and center (dashed) lines, course of the velocity angle $\gamma$ (solid line) and reference velocity angle $\gamma_\text{ref}$ (dashed). Employed wing: 9 m$^2$. Guidance parameters: $\theta_+=\theta_-=0.35\,$rad, $\phi_-=-0.2\,$rad, $\phi_+=0.2\,$rad; $\omega_\gamma=0.25\,$Hz.}
  \label{fig:PathForceGamma}
 \end{center}
\end{figure}
as well as the line forces and the velocity angle of the wing.
In order to evaluate the robustness of the velocity angle control loop, we checked the condition \eqref{eqn:stability_cond} for the following ranges
of the involved parameters: $|\vec{v}|\in[2,\,80]\,$m/s, $E_{eq}\in[2,\,8]$, $C_L\in[0.4,\,1]$, $A\in[6,\,12]\,$m$^2$, $d_s\in[1.8,\,3.1]\,$m, $m\in[1.7,\,3]\,$kg. The considered ranges of wing speed and equivalent
efficiency correspond to wind speeds in the interval $[1,\,10]\,$m/s, which covers a wide range of normal operating conditions
for airborne wind energy generators. With the considered intervals, the control systems results to be stable with the chosen value of the
gain $K_\text{c}=0.046\,$m/rad, thus indicating a good robustness of the approach. Indeed in our tests the same value of $K_c$ was used for all wind conditions and all three wings, with good results. We remark that, as anticipated in section \ref{SS:inner},  the saturation $\overline{\delta}_\text{m}=0.35\,$m on the motor position was never active since, in practice, the commanded position is much lower, of the order of 0.1$\,$m (see Fig. \ref{F:innermost_loop_perf}). The wing can be effectively steered with such small values of control input thanks to the presence of the geometric input $\delta_\text{g}$, whose effect is to contribute to steer the wing towards the center of the wind-window, hence facilitating the desired up-loops figure eights. In other words, the geometric input already gives place to a ``self-steering'' behavior, and the aim of the controller is to apply relatively slight corrections to prevent instability and divergence of the flown paths, especially in the middle of the wind window, where the geometric input is small and the wing, without feedback control, would head straight towards the ground. This can be clearly seen by comparing the numbered points in Figs. \ref{F:innermost_loop_perf} and \ref{fig:PathForceGamma} (top). On the bottom of Fig. \ref{fig:PathForceGamma}, a typical
example of the courses of the reference $\gamma_\text{ref}(t)$ and of the actual velocity angle $\gamma(t)$ obtained during the experiments is shown: it can be
noted that the middle control loop achieves good performance in tracking the desired velocity angle. The proposed control approach has been successfully tested under various conditions. The controller was able to deal
with varying wind speed between about $2\,$m/s up to $6\,$m/s, achieving similar, consistent flight paths also in the presence of gusts. Wind speeds lower than 2$\,$m/s would not allow the employed wings to fly without stalling, and we avoided to test with wind speeds higher than 6$\,$m/s not to stress too much the wing itself, the lines and the mechanical frame, pulleys and other components of  the prototype  (we recall that the involved forces increase linearly with the square of the wind speed). Misalignments of the GU with respect to the wind direction up to
$30^\circ$ did not pose a problem for the overall control strategy. However, if the target points (and thus the figure-eight paths) are not centered
with respect to the wind, the flight trajectory becomes slightly asymmetric in terms of altitude. In Fig. \ref{fig:10LoopOverlay}, this phenomenon can be seen in the $(\phi,\theta)$ plane. The forces are larger on one side of the loop and the wing tends to gain more altitude during these turns. Fig. \ref{fig:10LoopOverlay3D}
\begin{figure}[htb]
 \begin{center}
  \includegraphics[bbllx=13mm,bblly=68mm,bburx=195mm,bbury=205mm,width=8cm,clip]{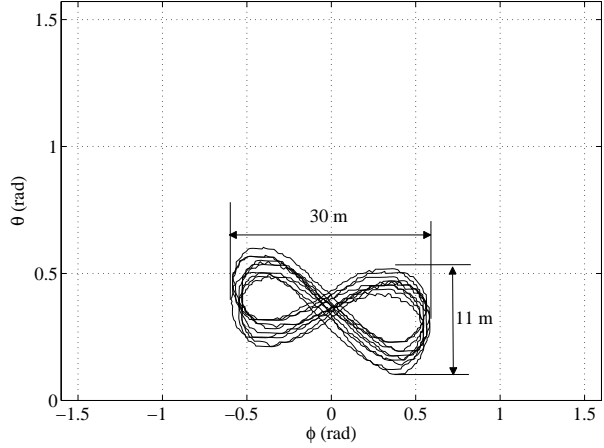}
  \caption{Experimental results. Ten consecutive figure-eight paths in the $(\phi,\theta)$ plane, and corresponding linear dimensions with the employed line length of 30 m. Employed wing: 9 m$^2$. Guidance parameters: $\theta_+=\theta_-=0.35\,$rad, $\phi_-=-0.2\,$rad, $\phi_+=0.2\,$rad; $\omega_\gamma=0.25\,$Hz.}
  \label{fig:10LoopOverlay}
 \end{center}
\end{figure}
\begin{figure}[htb]
 \begin{center}
  \includegraphics[bbllx=15mm,bblly=74mm,bburx=202mm,bbury=198mm,width=8cm,clip]{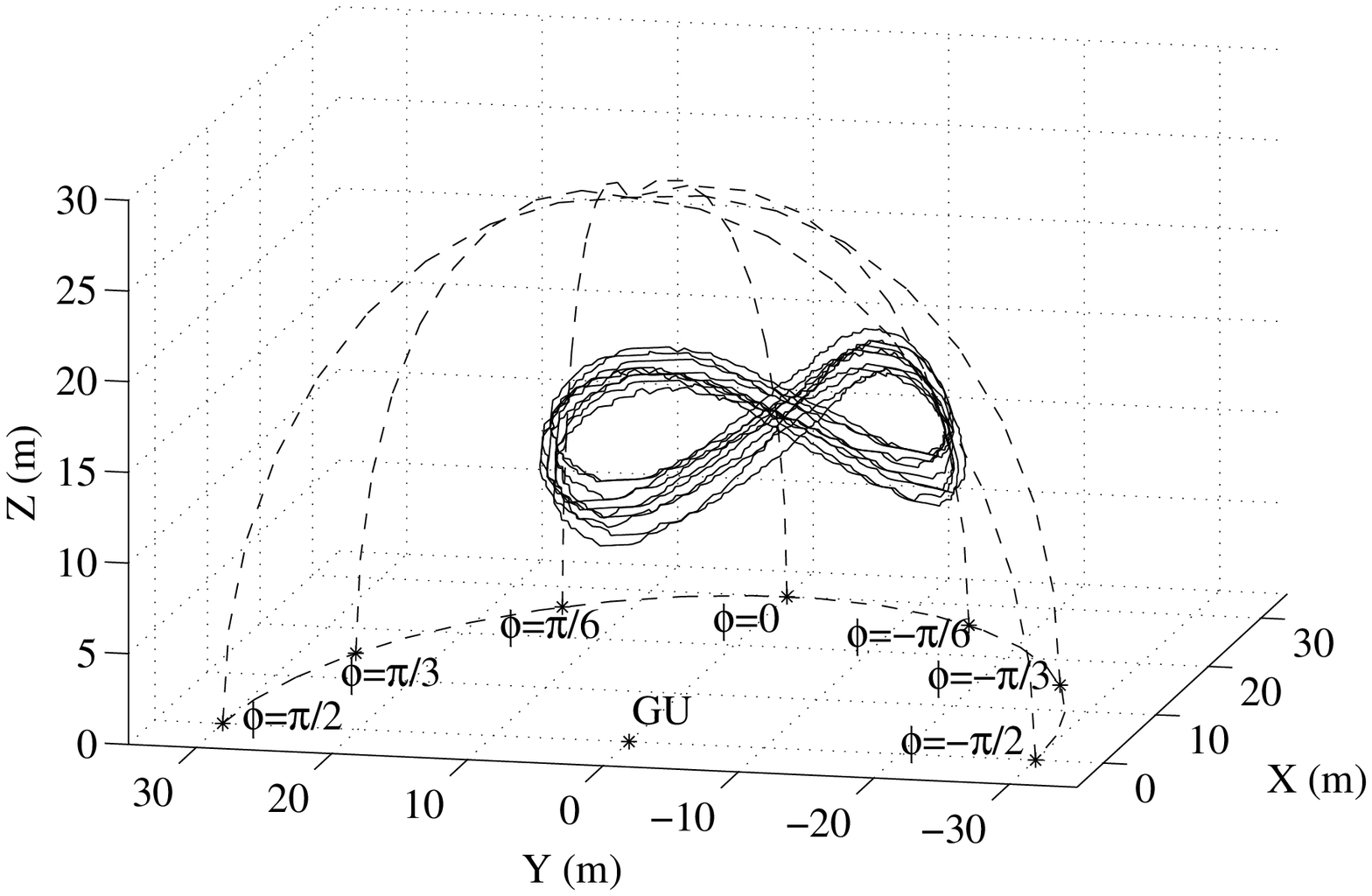}
  \caption{Experimental results. Ten consecutive figure-eight paths in the $(X,Y,Z)$ frame (solid lines). The dashed lines represent the wind window projected on the $(X,Y)$ plane as well as points with constant $\phi$ values and $\theta$ spanning the interval $[0,\pi/2]$. Employed wing: 9 m$^2$. Guidance parameters: $\theta_+=\theta_-=0.35\,$rad, $\phi_-=-0.2\,$rad, $\phi_+=0.2\,$rad; $\omega_\gamma=0.25\,$Hz.}
  \label{fig:10LoopOverlay3D}
 \end{center}
\end{figure}
shows the same paths as Fig. \ref{fig:10LoopOverlay}, but in the $(X,Y,Z)$ frame. Fig. \ref{fig:PathDiffWind} shows the results obtained with different wind speeds, again with the 9-m$^2$ wing and target points set to $\theta_+=\theta_-=0.8,\,$rad, $\phi_-=-0.4\,$rad, $\phi_+=0.4\,$rad.
\begin{figure}[htb]
 \begin{center}
  \includegraphics[bbllx=13mm,bblly=68mm,bburx=195mm,bbury=205mm,width=8cm,clip]{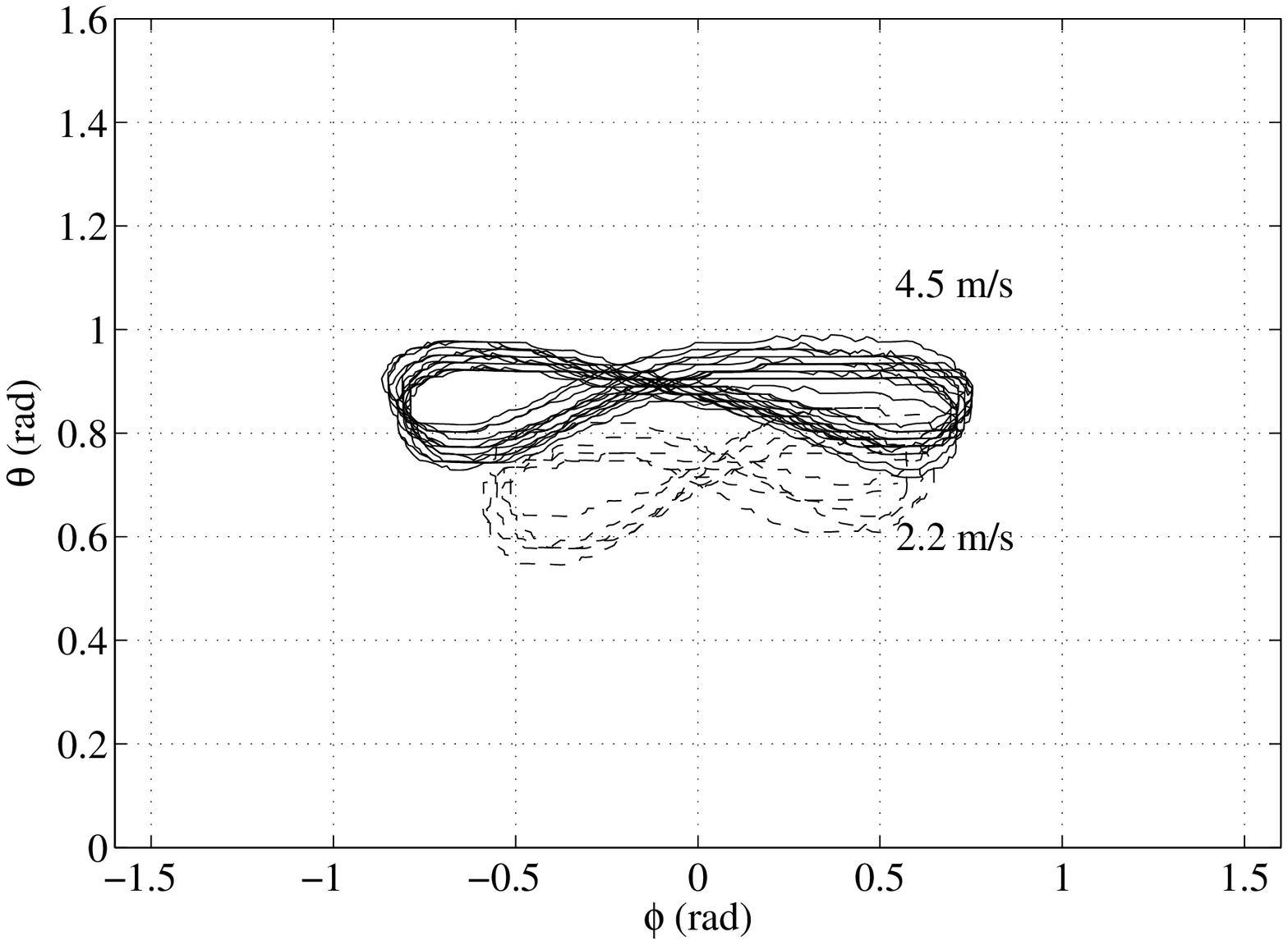}
  \caption{Experimental results.  Figure-eight paths flown with different wind speeds. The solid line corresponds to a series of consecutive figure-eights flown with average wind speed of $4.5\,$m/s, the dashed line  with average wind speed of $2.2\,$m/s. Employed wing: 9 m$^2$. Guidance parameters: $\theta_+=\theta_-=0.8\,$rad, $\phi_-=-0.4\,$rad, $\phi_+=0.4\,$rad; $\omega_\gamma=0.25\,$Hz.}
  \label{fig:PathDiffWind}
 \end{center}
\end{figure}
The  results of Fig. \ref{fig:PathDiffWind} highlight the fact that, since the proposed controller does not aim to track a given, specific reference path, the resulting figure-eight trajectories change with different wind conditions. However, the results also indicate that such changes are not dramatic, moreover the flown paths can be intuitively adjusted by changing the position of the target points $P_-,\,P_+$ employed in our control approach.\\
As mentioned in section \ref{SS:outer}, the cutoff frequency $\omega_\gamma$ can be used to influence the course of the reference velocity angle. The higher this frequency, the faster the transient behavior of $\gamma_\text{ref}(t)$ when the target point is switched, with consequent sharper turns and smaller flown paths: an example can be seen in Fig. \ref{fig:PathDiffGammaFilt}, where the target points are set to $\theta_+=\theta_-=0.55,\,\phi_-=-0.2,\,\phi_+=0.2$ and two different values of $\omega_\gamma$ are considered.
\begin{figure}[htb]
 \begin{center}
  \includegraphics[bbllx=13mm,bblly=68mm,bburx=195mm,bbury=205mm,width=8cm,clip]{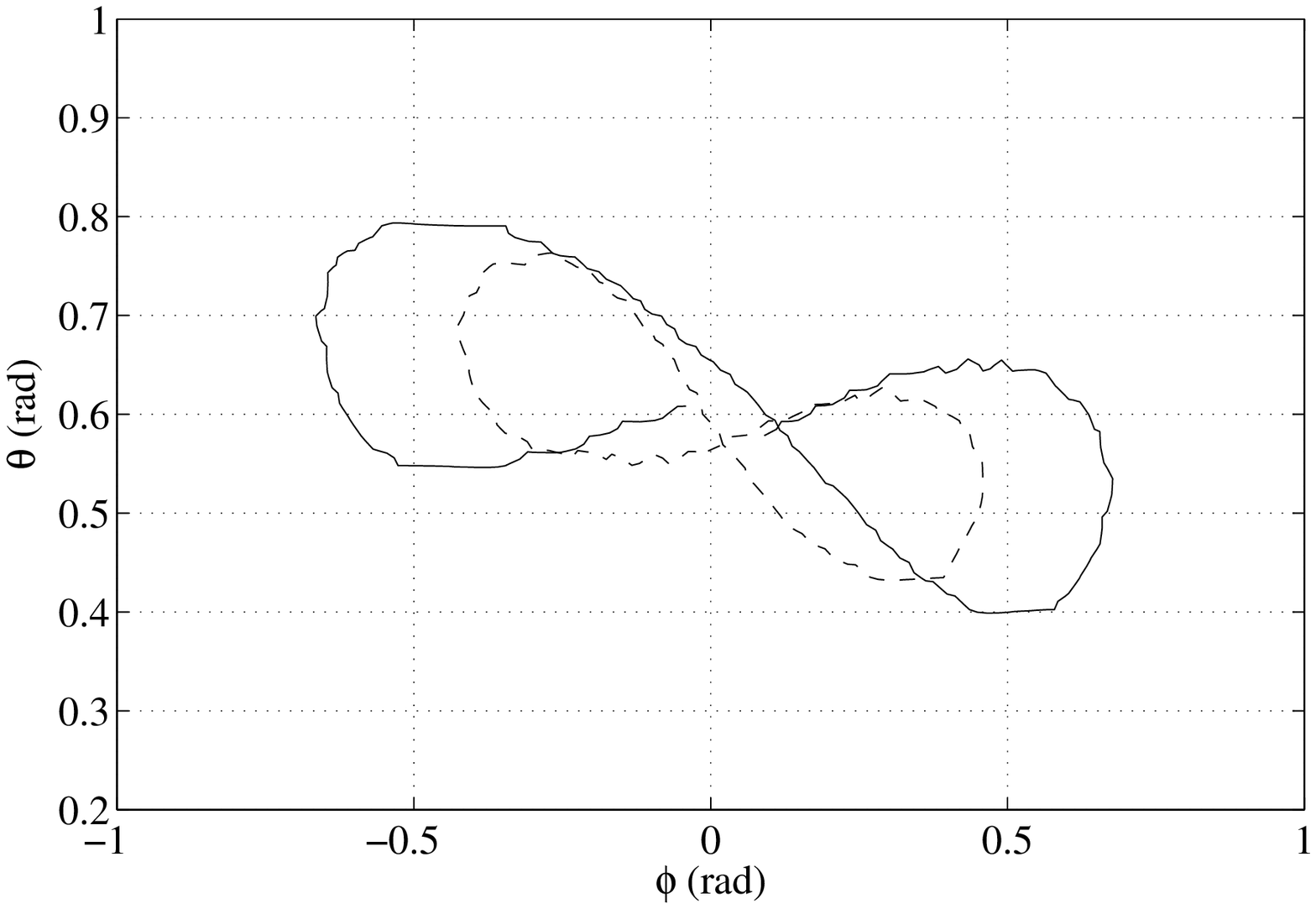}
  \caption{Experimental results. Figure-eight paths obtained with different values of the cutoff frequency $\omega_\gamma$. The solid line corresponds to $\omega_\gamma=0.25\,$Hz, the dashed one to $\omega_\gamma=1\,$Hz.  Employed wing: 9 m$^2$. Guidance parameters: $\theta_+=\theta_-=0.55\,$rad, $\phi_-=-0.2\,$rad, $\phi_+=0.2\,$rad.}
  \label{fig:PathDiffGammaFilt}
 \end{center}
\end{figure}
Figs. \ref{fig:PathForceGamma12} and \ref{fig:PathForceGamma6} show an analysis of typical figure-eight paths obtained with the 12-m$^2$ and with the 6-m$^2$ wings, respectively: while the qualitative results are similar to those obtained with the 9-m$^2$ wing, these Figures show the quantitative differences in terms of generated forces and behavior of the velocity angle. Finally, we present a comparison of the forces generated during extensive experimental tests with the different wings, and the corresponding theoretical values obtained e.g. from the results of \cite{FaMP11}. In particular, the generated force is expected to be affine in the quantity $E_{eq}\left(1+\frac{1}{E_{eq}^2}\right)^{3/2}\left(\,\cos{\theta}\cos(\phi)|\vec{W}|\right)^2$, with the gain being equal to $\frac{\rho\,C_L\,A}{2}$. The obtained results, shown in Figs. \ref{fig:Force}(a)-(c), indicate that the qualitative behavior is indeed consistent with the theory, albeit with some variability. The latter is mainly due to the uncertainty in the measurement of the wind speed at the wing's height and in the estimate of the lift coefficient and equivalent efficiency.
\begin{figure}[htb]
 \begin{center}
  \includegraphics[bbllx=7mm,bblly=43mm,bburx=200mm,bbury=236mm,width=8cm,clip]{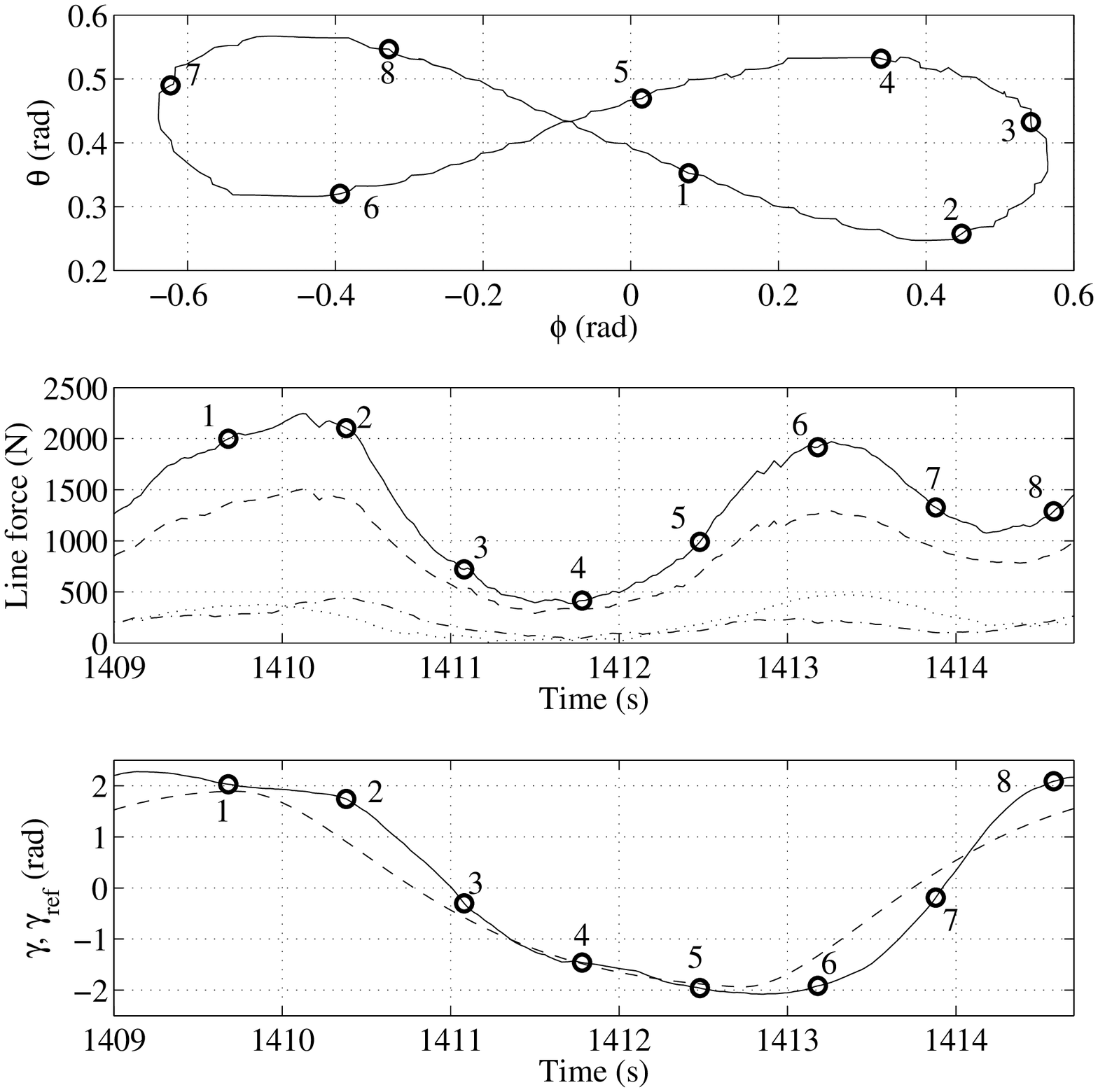}
  \caption{Experimental results. Single figure-eight path obtained during automatic test flights with about $3.1\,$m/s wind speed. From top to bottom: flying path in $(\phi,\theta)$ coordinates, course of the total force acting on the lines (solid line) and of the forces acting on the left (dotted), right (dash-dot) and center (dashed) lines, course of the velocity angle $\gamma$ (solid line) and reference velocity angle $\gamma_\text{ref}$ (dashed). Employed wing: 12 m$^2$. Guidance parameters: $\theta_+=\theta_-=0.4\,$rad, $\phi_-=-0.2\,$rad, $\phi_+=0.2\,$rad; $\omega_\gamma=0.25\,$Hz.}
  \label{fig:PathForceGamma12}
 \end{center}
\end{figure}
\begin{figure}[htb]
 \begin{center}
  \includegraphics[bbllx=10mm,bblly=43mm,bburx=200mm,bbury=236mm,width=8cm,clip]{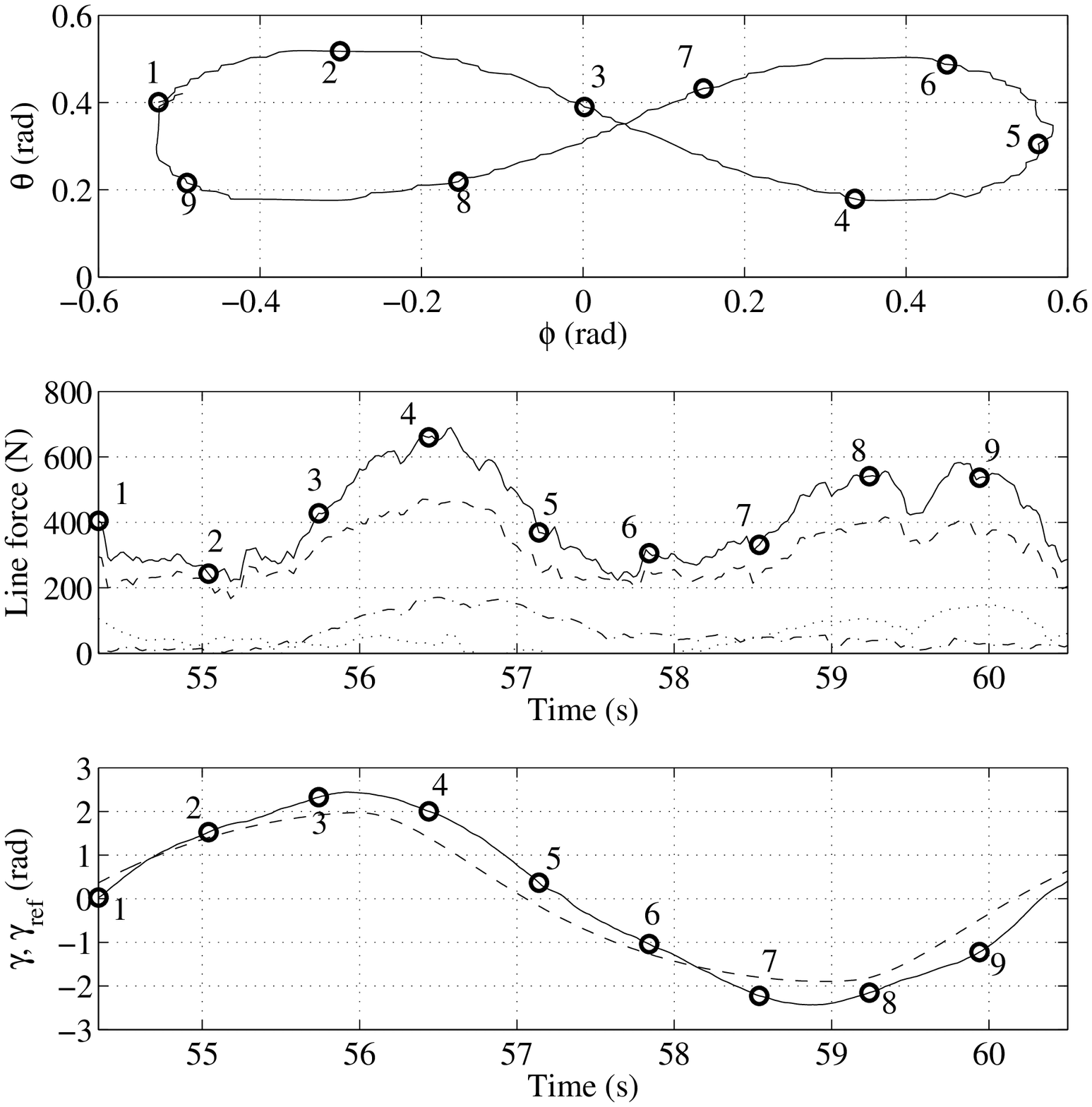}
  \caption{Experimental results. Single figure-eight path obtained during automatic test flights with about $2.2\,$m/s wind speed. From top to bottom: flying path in $(\phi,\theta)$ coordinates, course of the total force acting on the lines (solid line) and of the forces acting on the left (dotted), right (dash-dot) and center (dashed) lines, course of the velocity angle $\gamma$ (solid line) and reference velocity angle $\gamma_\text{ref}$ (dashed). Employed wing: 6 m$^2$. Guidance parameters: $\theta_+=\theta_-=0.35\,$rad, $\phi_-=-0.2\,$rad, $\phi_+=0.2\,$rad; $\omega_\gamma=0.25\,$Hz.}
  \label{fig:PathForceGamma6}
 \end{center}
\end{figure}
\begin{figure}[!hbt]
\centerline{
\begin{tabular}{c}
(a)\\
\includegraphics[bbllx=17mm,bblly=58mm,bburx=200mm,bbury=203mm,width=7cm,clip]{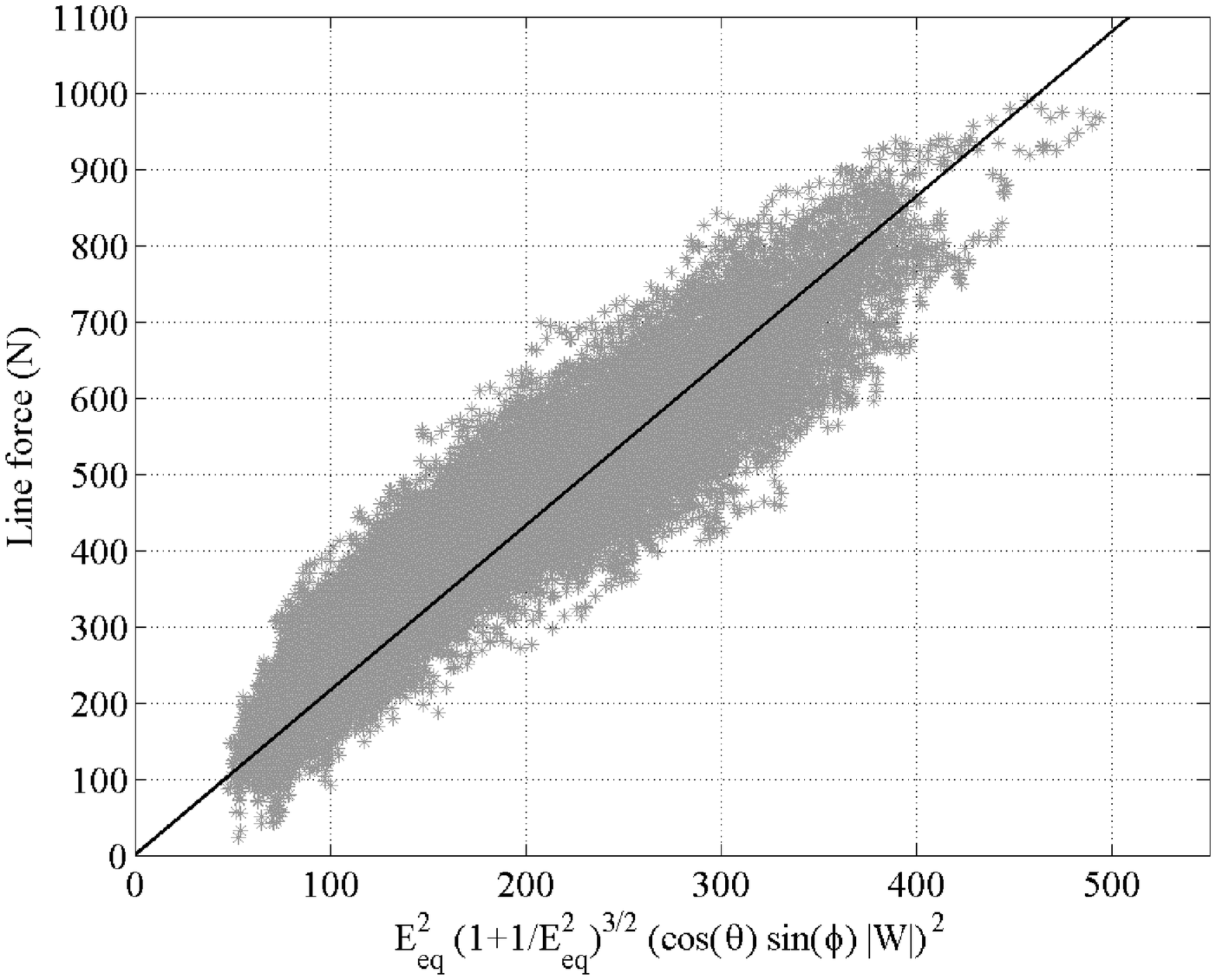}\\
(b)\\
\includegraphics[bbllx=17mm,bblly=58mm,bburx=200mm,bbury=203mm,width=7cm,clip]{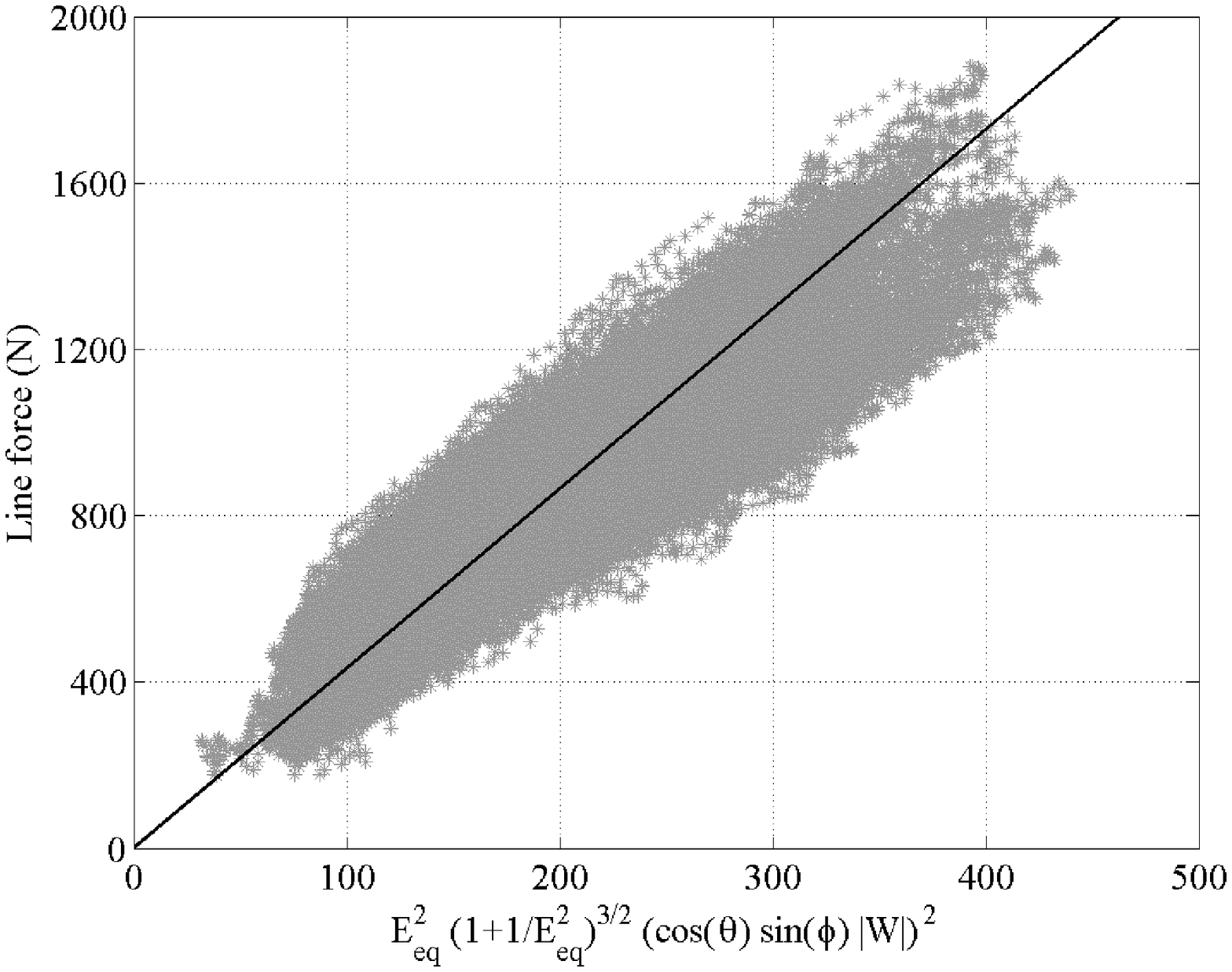}\\
(c)\\
\includegraphics[bbllx=17mm,bblly=58mm,bburx=200mm,bbury=203mm,width=7cm,clip]{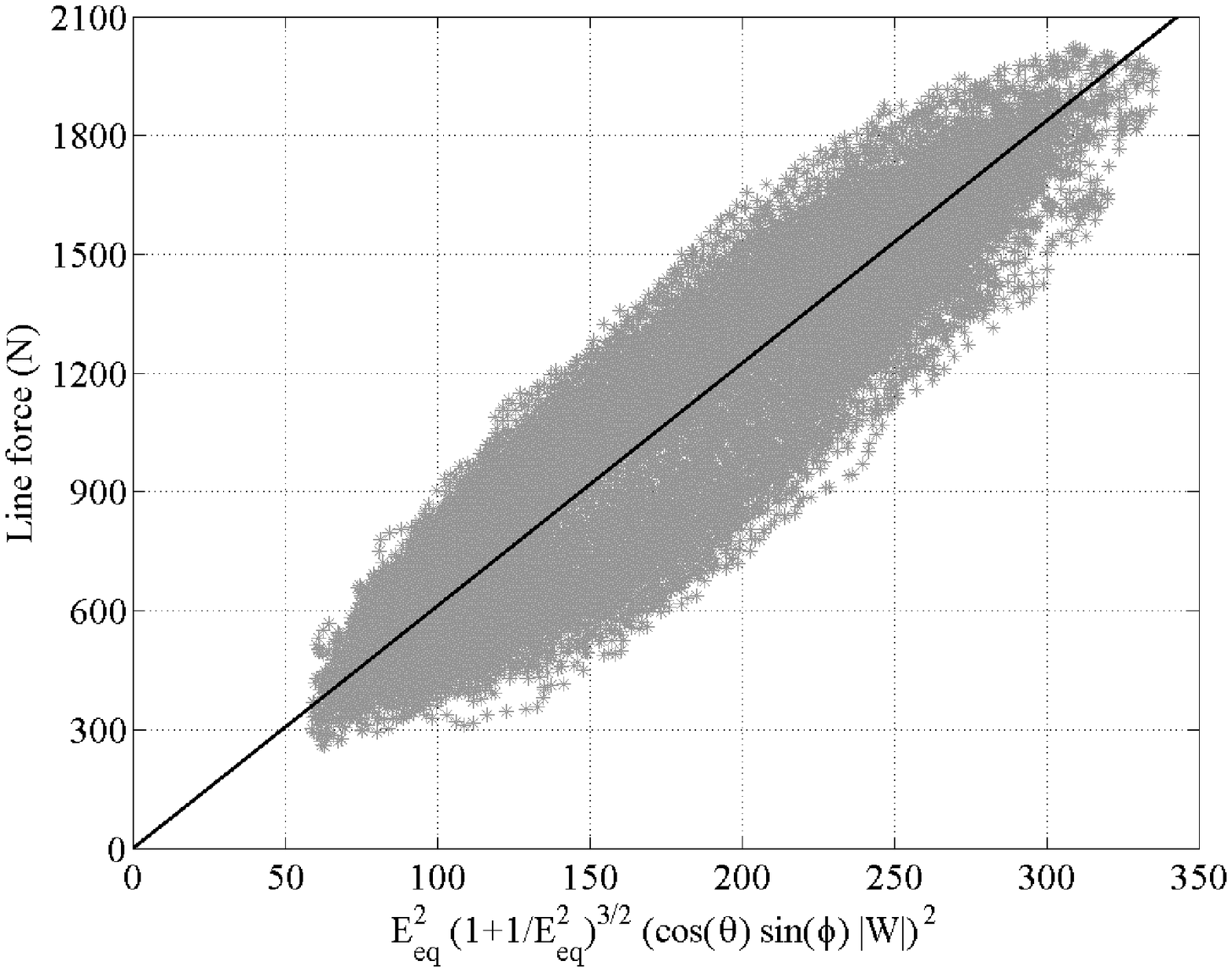}\\
\end{tabular}} \caption{Experimental results. Comparison between the measured values of total line force as a function of the quantity $E_{eq}\left(1+\frac{1}{E_{eq}^2}\right)^{3/2}\left(\,\cos{\theta}\cos(\phi)|\vec{W}|\right)^2$ (gray dots) and the theoretical linear relationship given by the gain $\frac{\rho\,C_L\,A}{2}$ as per the results of \cite{FaMP11} (solid lines).  (a) Airush One$^\circledR$ 6 kite, (b) Airush One$^\circledR$ 9 kite, (a) Airush One$^\circledR$ 12 kite. The lumped parameters for the kites are reported in Table \ref{T:wing_param}.} \label{fig:Force}
\end{figure}

\section{Conclusions}\label{S:conclusion}
A simplified model for the steering behavior of tethered wings in crosswind flight has been derived, and  its validity has been assessed with experimental data, collected with three different wings. Then, an approach to design a feedback controller for tethered wings has been proposed, with the aim to obtain figure-eight crosswind flying paths, to be used in airborne wind energy generators. The controller features three hierarchical levels. Differently from existing approaches in the literature, neither the measurement of the wind speed at the wing's altitude nor that of the effective wind speed are required. The control system involves few  parameters, that can be easily tuned. Moreover, a robustness analysis of the inner control loop has been carried out. The effectiveness of the approach has been shown through extensive experimental results obtained with a small scale prototype and different wings. Finally, a comparison between the forces measured during the tests and the ones predicted by the existing theoretical results has been carried out, showing a general qualitative consistency between theory and experiments.\\
The natural development of this research will be the execution of full generating cycles with the considered concept of airborne wind generator, and the comparison of the obtained results with the existing theoretical and numerical analyses. In order to achieve this goal, an additional motor/generator has to be used, and the related control algorithms have to be designed. These tasks could not be pursued with the considered prototype, which does not have energy generation capability, and they are subject of current research.

\appendix
\emph{Proof of Proposition \ref{P:gamma_dot}. }\small
For the sake of simplicity of notation, in the following we drop the dependance of time-varying variables from $t$. We start by deriving an analytic expression for the velocity angle rate. For this we take the derivative of (\ref{eqn:Gamma}) to get
\begin{equation}
 \dot{\gamma}= \dfrac{\cos{(\theta)}\,\dot{\theta}\ddot{\phi}-\sin{(\theta)}\,\dot{\phi}\dot{\theta}^2-\cos{(\theta)}\,\dot{\phi}\ddot{\theta}}				     {\cos{(\theta)}^2\,\dot{\phi}^2+\dot{\theta}^2}\label{eqn:GammaDot}
\end{equation}
The accelerations $\ddot{\theta}$ and $\ddot{\phi}$ in \eqref{eqn:GammaDot} can be expressed as functions of the forces acting on the wing by using the model equations \eqref{eqn:EqnMot1}-\eqref{eqn:Beta}.
By Assumption \ref{as:Crosswind}, neglecting all forces except for the aerodynamic ones and considering the balance of the lift and drag forces in the direction of the wing velocity $\vec{v}$ we have (see e.g. \cite{FaMP11}):
\begin{equation}
 \frac{\sin{\left(\Delta\alpha\right)}}{\cos{\left(\Delta\alpha\right)}} = \frac{1}{E_{eq}}\doteq\frac{C_{D,eq}}{C_L}\label{eqn:AOAcomplex}
\end{equation}
where $E_{eq}$ is the equivalent efficiency of the wing.
By the equation above we can see that $\Delta\alpha$ is small for a reasonable wing efficiency of 4-6. Measurements of test flights
have shown that $\Delta\alpha < 0.3\,$rad, and most of times $\Delta\alpha < 0.2\,$rad. Therefore we can linearize \eqref{eqn:AOAcomplex} to get
\begin{equation}
 \Delta\alpha = \frac{1}{E_{eq}}.\label{eqn:AOAsimple}
\end{equation}
Moreover, on the basis of Assumption \ref{as:SmallInput} we can simplify also $\eta$ in \eqref{eqn:eta} as $\eta = \Delta\alpha\psi$. Finally, by Assumption \ref{as:Crosswind} the component of absolute wind perpendicular to the tether is zero, and the velocity angle $\gamma\simeq\xi$.
Hence,  the aerodynamic force components given by \eqref{E:aero_force} and \eqref{E:aero_vectors} can be simplified as follows:
\begin{equation}\label{eqn:FaerNED}\begin{array}{r}
{}_L\vec{F}_\text{a}=\frac{1}{2}\rho C_LA|\vec{W}_e|^2\begin{pmatrix}
                          \Delta\alpha\cos{(\gamma)}-(\Delta\alpha^2\psi+\psi)\sin{(\gamma)}\\
	                  \Delta\alpha\sin{(\gamma)}+(\Delta\alpha^2\psi+\psi)\cos{(\gamma)}\\
	                   -1
	                 \end{pmatrix}+\\\frac{1}{2}\rho C_{D,eq}
                               A|\vec{W}_e|^2\begin{pmatrix}
		          -\cos{(\gamma)}\\
		          -\sin{(\gamma)}\\
		          -\Delta\alpha
	                 \end{pmatrix}
\end{array}
\end{equation}
By using (\ref{eqn:GammaDot}) and \eqref{eqn:EqnMot1}-\eqref{eqn:EqnMot2} with (\ref{eqn:FaerNED}) we obtain:
\begin{equation}\label{eqn:GammaDotSimple}
\begin{array}{r}
 \dot{\gamma} = \dfrac{\frac{1}{2}\rho C_LA|\vec{W}_e|^2}{rm}\dfrac{\dot{\theta}\cos{\gamma}+\cos{\theta}\dot{\phi}\sin{\gamma}}{\cos{\theta}^2\dot{\phi}^2+\dot{\theta}^2}
                        \left(\Delta\alpha^2+1\right)\psi\\
               +\dfrac{\frac{1}{2}\rho A|\vec{W}_e|^2\left( C_L
               \Delta\alpha-
               C_{D,eq}\right)\left(\dot{\theta}\sin{\gamma}-\cos{\theta}\dot{\phi}\cos{\gamma}\right)}
                       {rm\left(\cos{\theta}^2\dot{\phi}^2+\dot{\theta}^2\right)}\\
               +\dfrac{rm\sin{\theta}\dot{\phi}\left(\dot{\theta}^2+\cos{\theta}^2\dot{\phi}^2\right)}
                       {rm\left(\cos{\theta}^2\dot{\phi}^2+\dot{\theta}^2\right)}
              +\dfrac{g\cos{\theta}^2\dot{\phi}}{r\left(\cos{\theta}^2\dot{\phi}^2+\dot{\theta}^2\right)}
\end{array}
\end{equation}
Following the same argument by which \eqref{eqn:AOAsimple} is derived, we also obtain:
\begin{equation}\label{eqn:v_p_eff}
|\vec{W}_e(t)|^2\simeq \left(1+\frac{1}{E_{eq}^2}\right)^2|\vec{v}|^2.
\end{equation}
Considering \eqref{eqn:vP}, $|\vec{v}|$ is computed as:
\begin{equation}
 |\vec{v}|= \sqrt{r^2(\cos{(\theta)}^2\dot{\phi}^2+\dot{\theta}^2)},\label{eqn:vPP}
\end{equation}
which is, by the definition of $\gamma$ \eqref{eqn:Gamma}, also equal to
\begin{equation}
|\vec{v}|= r\dot{\theta}\cos{(\gamma)}+r\cos{(\theta)}\,\dot{\phi}\sin{(\gamma)}.\label{eqn:VPPgam}
\end{equation}
Finally, by combining (\ref{eqn:GammaDotSimple}) with (\ref{eqn:AOAsimple}) and \eqref{eqn:v_p_eff}-\eqref{eqn:VPPgam}, and considering the linearization of \eqref{eqn:psi} by Assumption \ref{as:SmallInput}, we get our result:
%
\begin{equation}\begin{array}{rcl}
 \dot{\gamma} &=&  \frac{C_L\rho A}{2md_s}\left(1+\frac{1}{E_{eq}^2}\right)^2|\vec{v}|\delta
                  +\frac{g\cos{\theta}\sin{\gamma}}{|\vec{v}|}+\sin{\theta}\dot{\phi}.\nonumber
\end{array}
\end{equation}

\bibliographystyle{IEEEtran}

\end{document}